\documentclass[12pt]{article}
\pdfoutput=1

\usepackage{draft,hyperref,tikz,cancel,subfig,float}
\usepackage{multirow}
\usepackage{soul}
\usepackage[nosort]{cite}
\usepackage{amsmath,mathtools}
\usepackage{amsfonts}
\usepackage{amsthm}
\usepackage{slashed}
\usepackage{xcolor}
\usepackage{graphicx}

\usetikzlibrary{snakes}
\usetikzlibrary{shapes.misc}
    \newcommand{\secref}[1]{\S\ref{#1}}
    \newcommand{\figref}[1]{Figure~\ref{#1}}

    \def\ie{\begin{equation}\begin{aligned}}
    \def\fe{\end{aligned}\end{equation}}

    \newcommand{\E}{{\epsilon}}

    \newcommand{\bP}{{\mathbb P}}

    \newcommand{\ld}{\lambda}

    \newcommand{\cT}{{\mathcal T}}
    
    \newcommand{\cZ}{{\mathcal Z}}

\theoremstyle{plain}
  
  \newtheorem{proposition}{Proposition}
  \newtheorem{lemma}{Lemma}
  
  \newtheorem{conjecture}{Conjecture}

\theoremstyle{definition}

\theoremstyle{remark}
  
  \newtheorem{remark}{Remark}

\begin{document}

\begin{titlepage}

\begin{center}

\hfill \\
\hfill \\
\vskip 1cm

\title{Evidence for an algebra of $G_2$ instantons II}

\author{Jihwan Oh$^{\dagger}$, Yehao Zhou$^{\ast}$}
\address{$^{\dagger}$ Mathematical Institute, University of Oxford, Woodstock Road, Oxford, OX2 6GG, United Kingdom}
\address{$^{\ast}$Perimeter Institute for Theoretical Physics, 31 Caroline St. N., Waterloo, ON N2L 2Y5, Canada}

\end{center}

\abstract{Building on our previous work [2109.01110], we will compute a new kind of $G_2$ instanton partition function. By doing so, we complete a set of building blocks of the instanton partition function associated with a large class of $G_2$ manifolds.} 
\end{titlepage}

\tableofcontents
\section{Introduction and summary}
This paper is a subsequent development of our previous work \cite{DelZotto:2021ydd} with Michele Del Zotto, where we gave an evidence for an algebra of $G_2$ instantons of a certain $G_2$ manifold, which is a cone on $\bC\bP^3$. In this work, we will generalize the technique developed in our previous work and apply it to another $G_2$ manifold, which is a self-dual 2-form bundle over $\bC\bP^2$, and compute an algebra of $G_2$ instantons associated to it. Combining with our previous result, we complete a set of building blocks of instanton partition function associated to a new class of non-compact $G_2$ manifold, as introduced in \cite{DelZotto:2021ydd}.

To be more precise about what we mean by the building blocks, let us recall how the new class of non-compact $G_2$ manifolds arises in the context of type IIA / M-theory duality. Consider stacks of D6 branes in type IIA string theory that wrap a set of special Lagrangian 3-cycles $\cC_i$ in a CY 3-fold $X$. The configuration can be uplifed to a hyperKähler fibration on a collection of associative
cycles $\cC_i$ and this gives rise to a $G_2$ manifold that is a U(1) fibration over the Calabi-Yau 3-fold $X$, where the U(1) fiber degenerates along the cycles $\cC_i$. Note that $\cC_i$ is at the locus of $\bC^2/\bZ_{N_i}$ singularities and the singularities enhance to $\bC^2/\bZ_{N_i+N_j+1}$
at the intersection of two stacks of D6-branes. Interesting physical information is encoded in the set of $\cC_i$'s and in the way they intersect each other. Therefore, we can focus on $\cC_i$'s and represent the $G_2$ manifold as a chain of $\cC_i$'s.

Consider the most basic chain, which consists of two $\cC_i=S^3$ intersecting at a point. Locally, the intersection can be described by two copies of $\bR^3$ intersecting at a point. This example was thoroughly analyzed in \cite{DelZotto:2021ydd}. Another elementary chain, which consists of three copies of $\cC_i$ that intersect at a point $\cC_1\cap\cC_2\cap\cC_3=\{\text{pt}\}$, is well known in the literature \cite{Atiyah:2001qf,Acharya:2001gy}. Locally, it is given by a bundle of self-dual 2-form on $\bC\bP^2$. Let us call it $\mathcal{X}$. The corresponding D6 brane stacks share 4 directions and intersect at a point in the remaining directions. 

The main result of this paper is the calculation of the instanton partition function of the $U(1)$ gauge theory on $\mathcal{X}$. Combining with our previous result, we have collected elementary building blocks that can be used to compute the instanton partition function of the new class of $G_2$ manifolds that we specified above. A natural question is then whether we can glue the elementary building blocks to compute instanton partition functions of more complicated $G_2$ manifolds. We also report partial progress in this direction in the last section.

Let us briefly outline the idea under the computation. The first step in computing the gauge theory partition function on a $G_2$ manifold is to use the saddle point approximation to convert the path integral into a sum of saddle points. We will call the saddle point configuration the $G_2$ instanton. For $\mathcal{X}$, we can explicitly construct the moduli space of the $G_2$ instantons by embedding the problem in the twisted M-theory \cite{Costello:2016mgj,Costello:2016nkh,Costello:2017fbo} and T-dualizing it into the type IIA frame. Along with the intersecting D6 brane system, the $G_2$ instanton saddles previously wrapped by the M2 branes become three stacks of the D2 branes. As a result, the previously intractable problem of the $G_2$ instanton moduli space translates into a tractable problem of the supersymmetric vacua of three stacks of D2 branes in the presence of three stacks of D6 branes. 

Once we have access to the moduli space of vacua, we can compute the $G_2$ instanton partition function that will be a key component of the entire Donaldson-Thomas partition function of the gauge theory on the $G_2$ manifold \cite{Iqbal:2003ds,Donaldson:1996kp,Acharya:1997gp}. Along the way, we will see the $G_2$ instanton moduli space $\cM(M_1,M_2,M_3)$ naturally admits an action of the algebra of operators of the 5d Chern-Simons theory associated with the $G_2$ manifold, where $M_i$ is the number of D2-branes in each of the three stacks. In other words, there is an action of $Y(\hat{\mathfrak{gl}}_1)\otimes Y(\hat{\mathfrak{gl}}_1)\otimes Y(\hat{\mathfrak{gl}}_1)$ on
\ie 
\bigoplus_{M_1,M_2,M_3}H^*_{\mathbf{T}_{\mathrm{edge}}}(\cM(M_1,M_2,M_3)).
\fe
Here, $Y(\hat{\mathfrak{gl}}_1)$ is affine Yangian of $\mathfrak{gl}_1$ and $\mathbf{T}_{\mathrm{edge}}$ is the torus action that regularizes noncompact $\cM$.

The outline of this paper is as follows. In \secref{sec:2}, we review the 7d/4d relation that we advocated in our previous work and explain how the current work completes the set of building blocks of the $G_2$ instanton partition function associated to the new class of $G_2$. In \secref{sec:3}, we study the moduli space of $G_2$ instantons by quantizing the strings stretched between the D-branes and constructing the low-energy effective action of 3d $\cN=2$ gauge theories on two D2-brane stacks. In \secref{sec:4}, we perform a formal analysis of the $G_2$ instanton moduli space. Importantly, the torus fixed point of the $G_2$ instanton moduli space is compact. This enables us to compute the equivariant K-theory index that counts the $G_2$ instantons. We then make contact with the 5d Chern-Simons theory. In \secref{sec:5}, we explain how to glue two copies of one of the two elementary building blocks of the $G_2$ instanton partition function. 
\section{7d/4d correspondence}\label{sec:2}
In our previous work \cite{DelZotto:2021ydd}, we have explained the following duality chain.
\begin{itemize}
    \item [\bf{IIB}] Consider N D3 branes probing a local singular Calabi-Yau cone X. Let us denote the worldvolume theory of the fractionalized D3 brane stack as $T^{4d}_{X,N}$. This is the 4d $\cN=1$ quiver superconformal field theory(SCFT).
    \item [\bf{IIA}] Using SYZ-like mirror symmetry \cite{Strominger:1996it}, one can arrive at the type IIA frame where D6 branes wrap a set of special Lagrangian 3-cycles $\cC_i$ in the mirror CY 3-fold $X^\vee$ \cite{Hanany:2001py,Feng:2005gw}. Each of D6 brane stack wrapping $\cC_i$ gives rise to the 4d $\cN=1$ vector multiplet and at the point where there is a nonzero intersection between $\cC_i$ and $\cC_j$, there is the 4d $\cN=1$ bifundamental chiral multiplet. 
    \item [\bf{M}] One can perform the $S^1$ uplift of the type IIA configuration to the M-theory on a new class of $G_2$ manifold $\Gamma_{X^\vee,N}$, which is a hyper-Kahler fibration over the set of $\cC_i$'s.
\end{itemize}
In view of the general philosophy of geometric engineering \cite{Leung:1997tw,DelZotto:2021gzy}, the following relation holds.
\ie
T^{4d}_{X,N}=T_{\Gamma_{X^\vee,N}},
\fe
where the latter theory is obtained by taking the geometric engineering limit of the M-theory on the $G_2$ manifold $\Gamma_{X^\vee,N}$. To be concrete, let us fix the manifold on which the theory $T^{4d}_{X,N}$ is supported as Taub-NUT manifold, $TN_1$. Then, the above relation implies that
\ie
\cZ^{4d}_{T_{X,N}}(TN_1)=\cZ_{\text{M-theory}}(TN_1\times \Gamma_{X^\vee,N}).
\fe
Recalling that the theory obtained from M-theory on $TN_k$ is a 7d $U(k)$ SYM, we can write the schematic form of the 4d/7d correspondence as
\ie
\cZ^{4d}_{T_{X,N}}(TN_1)=\cZ^{7d}_{U(1)}(\Gamma_{X^\vee,N})
\fe

As the second step {\bf{IIA}} indicates, the theory is entirely characterized by the set of special Lagrangian 3-cycles $\cC_i$ and their intersection matrix. The simplest example of a chain of 3-cycles is given by $\cC_1=S^3_1$ and $\cC_2=S^3_2$ where $S^3_1\cap S^3_2=\{pt\}$. The local model at the intersection of the D6 brane stacks is the 4d $\cN=1$ bifundamental chiral multiplet theory. At the same time, zooming to the intersection of the 3-cycles, one finds a well-known $G_2$ cone on $\bC\bP^3$ \cite{BryantSalamon}. We discussed this example in detail in \cite{DelZotto:2021ydd}. 

There is another simple local $G_2$ cone, which is a self-dual two-form bundle over $\bC\bP^2$. Similar to the above example, one can find this geometry by zooming in the local intersection of three copies of 3-spheres. In view of the 4d/7d relation, one might wonder what the corresponding 4d theory is. This theory is believed to be the superconformal field theory, which lives at the origin of the singular moduli space of a certain 4d $\cN=1$ field theory \cite{Atiyah:2001qf}.

At this point, looking at the two types of local intersections described above, one might notice a structural similarity of the 4d/7d relation and the 4d/2d relation in class-S construction \cite{Gaiotto:2009we,Gaiotto:2009hg}. For example, the class-S superconformal index can be computed by gluing the partition function of topological field theory on a pair of pants and a cylinder \cite{Gadde:2011ik}; the idea of gluing even goes back to the example of topological vertex \cite{Aganagic:2003db,Iqbal:2007ii}. If we admit that there exists a class of $G_2$ manifolds characterized by the chain of associative 3-cycles, we may wonder about an algorithmic way of computing the associated $G_2$ instanton partition function using some elementary building blocks. In our previous work, we studied the cylinder component of the general $G_2$ instanton partition function. This corresponds to the simple intersection between two 3-cycles. We proved that $T_{\Gamma_{X^\vee,N}}$ admits an action of two copies of the 5d CS algebra
\ie
U_{\E_1}(\text{Diff}_{\E_2}\bC\otimes\mathfrak{gl}_1)\times U_{\E'_1}(\text{Diff}_{\E_2}\bC\otimes\mathfrak{gl}_1),
\fe
where $Diff_\epsilon\bC$ is the algebra of differential operators on $\bC$, parametrized by $\epsilon$. 

In this work, we study the pair-of-pants component, which corresponds to a triple intersection of three 3-cycles. As one can naturally expect, $G_2$ instanton partition function admits an action of
\ie
U_{\E_1}(\text{Diff}_{\varepsilon}\bC\otimes\mathfrak{gl}_1)\times U_{\E_2}(\text{Diff}_{\varepsilon}\bC\otimes\mathfrak{gl}_1)\times U_{\E_3}(\text{Diff}_{\varepsilon}\bC\otimes\mathfrak{gl}_1),
\fe
where
\ie
\E_1+\E_2+\E_3=0.
\fe
We will explain the origin of the algebra in \secref{subsec:twistedM}.
\subsection{The $G_2$ manifold}
Let us briefly recall the $G_2$ manifold of our interest following \cite{Atiyah:2001qf}. There are two representations. First, it is a cone on $Y=SU(3)/U(1)^2$. It can also be represented as a bundle of self-dual 2-forms on $\bC\bP^2$, $\mathcal{X}$. An important point, which  we will use to go to the type IIA frame, is the fixed points under $S^1$ reduction are three copies of $\bR^3$, which intersect at a point in $\bR^6$. Moreover, 
\ie
\mathcal{X}/S^1\equiv\bR^6\text{ and }Y/S^1\equiv S^5.
\fe
The $U(1)$ fixed points of $Y$ is 
\ie
S^2\cup S^2\cup S^2.
\fe
\subsection{Twisted M-theory perspective}\label{subsec:twistedM}
Before we dive into the computation of the $G_2$ instanton partition, let us recall twisted M-theory set-up \cite{Costello:2016mgj,Costello:2016nkh,Costello:2017fbo} to get some intuition on the algebra that acts the $G_2$ instanton partition function, which we will get in the following section. To do so, we will focus on the local patch that we mentioned in the previous subsection.

The relevant twisted M-theory background with which we will work is $TN_1\times\mathcal{X}$. Since both of $\mathcal{X}$ and $TN_1$ admits the $U(1)$ action, we have two choices of circle reduction, and each of the reductions leads to a different type IIA configuration. Reducing $S^1_\mathcal{X}$ in $\mathcal{X}$, we get three stacks of D6-branes that extend over the direction $TN$ and intersect at a point in the remaining 6-directions. Let us call this type IIA background frame 1. This circle reduction is Q-exact in the twisted M-theory, and the resulting type IIA frame retains the same physical information as in the twisted M-theory. On the other hand, by reducing $S^1_{TN}$, one gets a D6-brane wrapping $\mathcal{X}$; we will call it frame 2. By $G_2$ instantons, we mean the instanton configuration of the 7d gauge theory associated to the D6-brane. As in \cite{DelZotto:2021ydd}, we will take an indirect way to compute the instanton partition function: we will model the $G_2$ instantons configuration as D2-branes embedded in the three different stacks of D6-branes in frame 1.  

Even without a detailed computation of the partition function, we can notice some general structure under it by studying the D6-brane stacks. First, recall that the 5d $U(N)$ Chern-Simons theory on $TN_1\times\bR_t$ is obtained from a worldvolume theory of N D6 branes on $\bR^2_{\Omega_{\E}}\times TN_1\times\bR$ after localization on $\bR^2_{\Omega_{\E}}$(an Omega deformed 2-plane) \cite{Costello:2018txb}. There are three stacks of $D6$ branes that wrap three copies of $\bR_i^3$, respectively, in frame 1. Each of the $\bR_i^3$ has $U(1)$ isometry and the overall $U(1)$ decouples, as indicated in \cite{Atiyah:2001qf}. In other words, one can use the three dependent $U(1)$ isometry to turn on the Omega background $\Omega_{\E_1}\times\Omega_{\E_2}\times\Omega_{\E_3}$, where $\E_1+\E_2+\E_3=0$. \footnote{The isometries of the original $G_2$ geometry is PSU(3), so we have an Omega background with 2 parameters corresponding to the Cartan of PSU(3). We thank Kevin Costello for letting us know about this point.} Since three stacks of D6 branes intersect at a point, the resulting three copies of the 5d CS theory on $TN_1\times\bR_i$ share the holomorphic direction $TN_1$ and the three topological lines $\bR_i(\equiv\bR^3_i/\Omega_{\E_i})$ intersect at a point.
\begin{figure}[H]
    \centering
    \vspace{-0.2cm}
    \includegraphics[width=6cm]{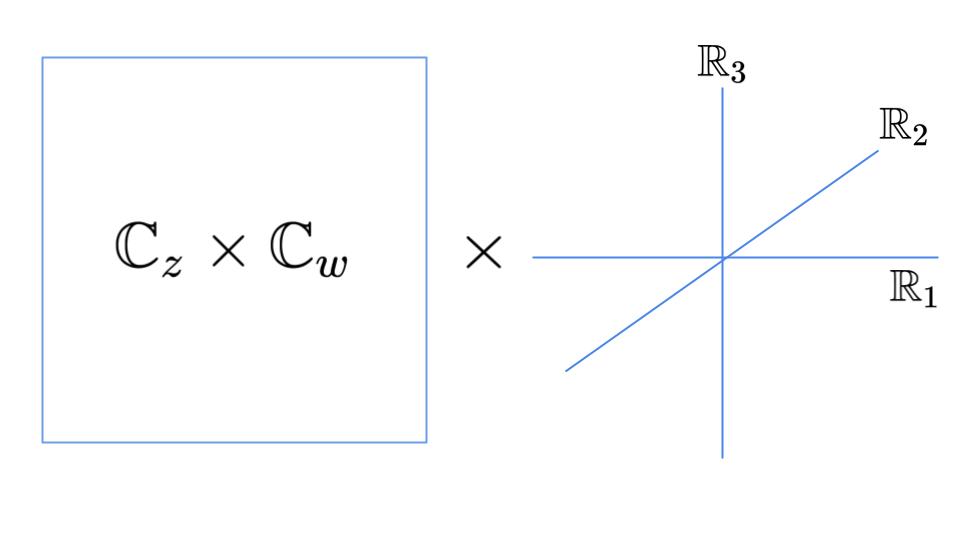}
    \vspace{-0.4cm}
    \caption{Three copies of the 5d CS theory}
    \vspace{-0.2cm}
    \label{fig:doublecontraction}
     \centering
\end{figure}

The associated algebra of operators of the 5d CS theory is 
\ie
U_{\E_1}(\text{Diff}_\varepsilon\bC\otimes\mathfrak{gl}_1)\times U_{\E_2}(\text{Diff}_\varepsilon\bC\otimes\mathfrak{gl}_1)\times U_{\E_3}(\text{Diff}_\varepsilon\bC\otimes\mathfrak{gl}_1),
\fe
where $\varepsilon$ is the constant that parameterizes the strength of the B field $\varepsilon d\bar zd\bar w$ on $TN_1$; this B field descends from the M-theory 3-form $C=\varepsilon V^\flat d\bar zd\bar w$ by reducing the circle $S^1$ on which the $U(1)$ vector field $V$ generates rotation. $V^\flat$ is the dual 1-form of $V$. In the following sections, especially in {\bf{Proposition 6}} and \secref{sec:connectionto5dcs}, we will find that the $G_2$ instanton partition function admits an action of the above algebra. 
\section{Moduli space analysis}\label{sec:3}
\subsection{From geometry to branes and field theory}
The $G_2$ manifold and the corresponding $G_2$ instantons of our interest map into the following D6/D2 stacks:
\ie
D6_i=\bR^3_i\times\bR_{0123}^4,\quad D2_i=\bR^3_i.
\fe
The D-brane stacks intersect in the following way:
\ie
D6_1\cap D6_2\cap D6_3=\bR^4_{0123},\quad D6_i\cap D2_i=\bR^3_i,\quad D6_i\cap D2_j=\{\text{pt}\},\quad D2_i\cap D2_j=\{\text{pt}\},
\fe
Since the M-theory/type IIA duality that we utilized to arrive at this D-brane configuration preserves the amount of supersymmetry and the M-theory compactification on $G_2$ manifold preserves 1/8 SUSY, the above D-brane configuration preserves 4 supercharges.

Our D-brane configuration can be thought of as an extended version of the configuration that we studied in our previous work \cite{DelZotto:2021ydd}. We will simply summarize the results here referring to \cite{DelZotto:2021ydd,Nekrasov:2016gud} for more details of the derivation.
\ie\label{fieldcontent}
D2_i-D2_i\text{ strings }&\equiv~\text{3d }\cN=2\text{ vector multiplets }\cV_2\text{ and chiral multiplets }\Phi_i,X_i,Y_j\\
D2_i-D2_j\text{ strings }&\equiv~\text{3d }\cN=2\text{ bifundamental hyper multiplets } S_{ij}\\
D2_i-D6_i\text{ strings }&\equiv~\text{3d }\cN=4\text{ hypermultiplet }I_i,J_i\\
D6_i-D6_i\text{ strings }&\equiv~\text{7d }U(N_i)\text{ vector multiplet }\cV^{7d}_i\\
D6_i-D6_j\text{ strings }&\equiv~\text{4d }\cN=1\text{ bifundamental chiral multiplet }\mathcal{C}_{ij}
\fe
For our purpose that will be described in the next subsection, we can focus on the strings that are associated with the 3d massless fields: $D2_i-D2_i$, $D2_i-D6_i$, $D2_i-D2_j$, $D2_i-D6_j$ strings. Considering only the massless states, we get 3d $\cN=2$ theory on the $D2$ branes.

The following figure summarizes the bosonic field content of the theory determined above.
\begin{figure}[H]
    \centering
    \includegraphics[width=15cm]{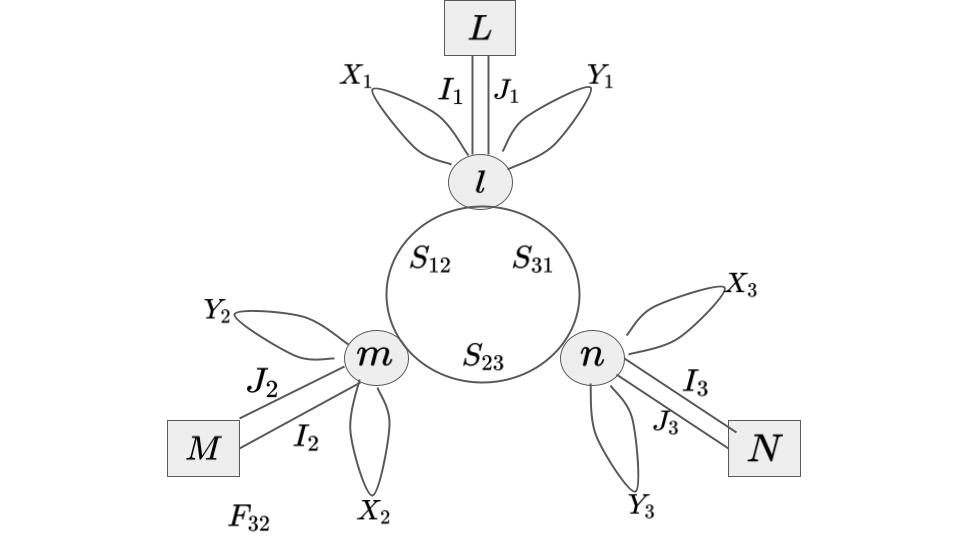}
    \caption{A collection of bosons}
    \label{fig:doublecontraction}
     \centering
\end{figure}
\subsection{Moduli space of instantons}\label{sec:Dterm}
As we have translated the problem of the moduli space of $G_2$ instantons $\cM$ to the moduli space of D2 branes lying in D6 branes, we need to compute the supersymmetric vacua of the worldvolume theory of the D2 brane determined in the previous subsection. 

In terms of 3d $\cN=2$ superfields, $\cM$ is parametrized by the scalar components of  
$I_i$, $J_i$, $X_i$, $Y_i$, $S_{ij}$. Repeating the exercise similar to \cite{DelZotto:2021ydd}(analyzing D-term equations), one can get the following defining equations of the moduli space.

From the first gauge node$(D2_1)$, we get
\ie\label{Dtermsp1}
&\left[X_1,X_1^\dagger\right]+\left[Y_1,Y_1^\dagger\right]+I_1I_1^\dagger-J_1J_1^\dagger+S_{12}^\dagger S_{12}+S_{31}^\dagger S_{31}-\xi\cdot\text{I}_{M_1\times M_1}=0.\\
&\left[X_1,Y_1\right]+I_1J_1=0.
\fe
We can do the same for the gauge nodes$(D2_2,D2_3)$ of the quiver. We get
\ie\label{Dtermsp2}
&\left[X_2,X_2^\dagger\right]+\left[Y_2,Y_2^\dagger\right]+I_2I_2^\dagger-J_2J_2^\dagger+S_{23}^\dagger S_{23}+S_{12}^\dagger S_{12}-\xi\cdot\text{I}_{M_2\times M_2}=0,\\
&\left[X_2,Y_2\right]+I_2J_2=0,\\
&\left[X_3,X_3^\dagger\right]+\left[Y_3,Y_3^\dagger\right]+I_3I_3^\dagger-J_3J_3^\dagger+S_{23}^\dagger S_{23}+S_{12}^\dagger S_{12}-\xi\cdot\text{I}_{M_3\times M_3}=0,\\
&\left[X_3,Y_3\right]+I_3J_3=0.
\fe

Hence, the moduli space of $G_2$ instantons $\cM^{L,M,N}_{M_1,M_2,M_3}$ is the $U(M_1)\otimes U(M_2)\otimes U(M_3)$ quotient of the space of solutions of \eqref{Dtermsp1}, \eqref{Dtermsp2}, where the groups $U(M_1)\otimes U(M_2)\otimes U(M_3)$ act as
\ie
\left(X_i,Y_i,I_i,J_i\right)&\mapsto\left(g_i^{-1}Xg_i,g_i^{-1}Yg_i,g_i^{-1}I,Jg_i\right),\quad\text{where }g_i\in U(M_i),\\
S_{ij}&\mapsto g^{-1}_iS_{ij}g_j,\quad\text{where }g_i\in U(M_i),~g_j\in U(M_j)
\fe
\section{Index}\label{sec:4}
\subsection{Moduli space as a Nakajima quiver variety}\label{sec:5.1}

Let us denote the corresponding Nakajima quiver variety of \figref{fig:doublecontraction}  by ${\cM}^{1,1,1}(M_1,M_2,M_3)$. ${\cM}^{1,1,1}(M_1,M_2,M_3)$ is smooth variety of dimension $M_1M_2+M_2M_3+M_3M_1+2M_1+2M_2+2M_3$, defined as the space of solutions to equations \eqref{Dtermsp1}, \eqref{Dtermsp2} modulo the action of $U(M_1)\times U(M_2)\times U(M_3)$.
\ie
S_{12}I_1&=S_{23}I_2=S_{31}I_3=0
\fe
Note that the indices are the node numbers. The three equations of \eqref{Dtermsp1}, \eqref{Dtermsp2} involving complex conjugate are the real moment map $\mu_{\bR}=(\xi\cdot\text{I}_{M_1\times M_1},\xi\cdot\text{I}_{M_2\times M_2},\xi\cdot\text{I}_{M_3\times M_3})$, and the other three equations are the complex moment map $\mu_{\bC}=0$. Based on the result of Kempf and Ness \cite{kempf1979length}, there is a holomorphic description of the moduli space ${\cM}^{1,1,1}(M_1,M_2,M_3)$.
\begin{lemma}\label{Lemma: algebro-goemetric description}
Assume that $\xi>0$, then we have an isomorphism:
\ie 
\cM^{1,1,1}(M_1,M_2,M_3)=\mu_{\bC}^{-1}(0)^{\mathrm{stable}}/\mathrm{GL}_{M_1}\times \mathrm{GL}_{M_2}\times \mathrm{GL}_{M_3},
\fe
where ``stable" means the locus inside $\mu_{\bC}^{-1}(0)$ such that $\{X_i,Y_i,I_i,J_i\}$ satisfies 
\ie\label{eqn: stability}
\bC\langle X,Y,S\rangle( \mathrm{Im}(I_1)+ \mathrm{Im}(I_2)+ \mathrm{Im}(I_3))&=\bC^{M_1}\oplus \bC^{M_2}\oplus \bC^{M_3}.
\fe
\end{lemma}
We can compare the definition of ${\cM}^{1,1,1}(M_1,M_2,M_3)$ with a more standard quiver variety denoted by $\widetilde{\cM}^{1,1,1}(M_1,M_2,M_3)$, where the quiver comes from adding three backward arrows $T_{21},T_{32},T_{13}$ to the quiver in \figref{fig:doublecontraction}, and the stability is that
\ie
\bC\langle X,Y,S,T\rangle( \mathrm{Im}(I_1)+ \mathrm{Im}(I_2)+ \mathrm{Im}(I_3))&=\bC^{M_1}\oplus \bC^{M_2}\oplus \bC^{M_3}.
\fe
Then it is easy to see that the natural embedding of the quiver data induces a closed embedding $\cM^{1,1,1}(M_1,M_2,M_3)\hookrightarrow\widetilde{\cM}^{1,1,1}(M_1,M_2,M_3)$ as the zero locus of the universal maps $\cT_{21},\cT_{32},\cT_{13}$, where $\cT_{ij}\in \mathrm{Hom}(\cV_i,\cV_j)$ is the map between the universal bundles on $\widetilde{\cM}^{1,1,1}(M_1,M_2,M_3)$. Using this embedding, we can show that the singularities of $\cM^{1,1,1}(M_1,M_2,M_3)$ are locally of complete intersection.

\begin{proposition}\label{prop: regular embedding}
Assume that $\xi\neq 0$, then the embedding $\cM^{1,1,1}(M_1,M_2,M_3)\hookrightarrow\widetilde{\cM}^{1,1,1}(M_1,M_2,M_3)$ is regular of codimension $M_1M_2+M_2M_3+M_3M_1$. Moreover, locally the set of matrix elements of $\mathcal T_{21},\mathcal T_{32},\mathcal T_{13}$ is a regular sequence of length $M_1M_2$, $M_2M_3$, $M_3M_1$, respectively.
\end{proposition}

\begin{proof}
If we prove that $\cM^{1,1,1}(M_1,M_2,M_3)$ is local complete intersection (l.c.i) of pure dimension $M_1M_2+M_2M_3+M_3M_1+2M_1+2M_2+2M_3$, then it follows that locally the set of matrix elements of $\mathcal T_{21},\mathcal T_{32},\mathcal T_{13}$ is a regular sequence of length $M_1M_2$, $M_2M_3$, $M_3M_1$, respectively. To this end, we need to show that $\mu_{\bC}^{-1}(0)^{\mathrm{stable}}$ is l.c.i of pure dimension $M_1M_2+M_2M_3+M_3M_1+2M_1+2M_2+2M_3+M_1^2+M_2^2+M_3^2$, since the quotient map $\mu_{\bC}^{-1}(0)^{\mathrm{stable}}\to \mu_{\bC}^{-1}(0)^{\mathrm{stable}}/\mathrm{GL}_{M_1}\times \mathrm{GL}_{M_2}\times \mathrm{GL}_{M_3}=\cM^{1,1,1}(M_1,M_2,M_3)$ is a principal $\mathrm{GL}_{M_1}\times \mathrm{GL}_{M_2}\times \mathrm{GL}_{M_3}$ bundle and l.c.i property descends to a smooth morphism. 

We claim that the map 
\ie\nonumber
\mu_{\bC}:& \mathbf{M_1}\times \mathbf{M_1}^*\times \mathrm{End}(\mathbf{M_1})^{2}\times \mathbf{M_2}\times \mathbf{M_2}^*\times \mathrm{End}(\mathbf{M_2})^{2}\times\mathbf{M_3}\times \mathbf{M_3}^*\times \mathrm{End}(\mathbf{M_3})^{2}\times\\
&\mathrm{Hom}(\mathbf{M_1},\mathbf{M_2})\times\mathrm{Hom}(\mathbf{M_2},\mathbf{M_3})\times\mathrm{Hom}(\mathbf{M_3},\mathbf{M_1})\to \mathfrak{gl}_{M_1}\times \mathfrak{gl}_{M_2}\times \mathfrak{gl}_{M_3}
\fe
is flat. Note that the component $\mathrm{Hom}(\mathbf{M_i},\mathbf{M_j})$ where $S_{ij}$ takes value in, plays no role in the map $\mu_{\bC}$, so $\mu_{\bC}$ factors through the product of moment maps $$\mu_{\bC}^1\times \mu_{\bC}^1\times \mu_{\bC}^1:\prod_{i=1}^3\left( \mathbf{M_i}\times \mathbf{M_i}^*\times \mathrm{End}(\mathbf{M_i})^{2}\right)\to\prod_{i=1}^3 \mathfrak{gl}_{M_i}.$$
Using the fact that $\mu_{\bC}^1:\mathbf{M_i}\times \mathbf{M_i}^*\times \mathrm{End}(\mathbf{M_i})^{2}\to \mathfrak{gl}_{M_i}$ is flat, using Crawley-Boevey's criterion on the flatness of moment map \cite[Theorem 1.1]{Crawley-Boevey:2001}. Hence $\mu^1_{\bC}\times\mu^1_{\bC}\times\mu^1_{\bC}$ is flat.

Since $\{0\}\hookrightarrow \mathfrak{gl}_{M_1}\times \mathfrak{gl}_{M_2}\times \mathfrak{gl}_{M_3}$ is a regular embedding, the flatness of $\mu_{\bC}$ implies that $\mu_{\bC}^{-1}(0)$ embeds in the ambient space regularly and has dimension
\ie
\dim\left[\prod_{i=1}^3\mathbf{M_i}\times \mathbf{M_i}^*\times \mathrm{End}(\mathbf{M_i})^{2}\times\prod_{1\leq i<j\leq3} \mathrm{Hom}(\mathbf{M_i},\mathbf{M_j})-\prod_{i=1}^3\mathfrak{gl}_{M_i}\right]\\
=M_1M_2+M_2M_3+M_3M_1+2M_1+2M_2+2M_3+M_1^2+M_2^2+M_3^2.
\fe
This shows that $\mu_{\bC}^{-1}(0)$ is l.c.i. of pure dimension $M_1M_2+M_2M_3+M_3M_1+2M_1+2M_2+2M_3+M_1^2+M_2^2+M_3^2$, and obviously its open subset $\mu_{\bC}^{-1}(0)^{\mathrm{stable}}$ has the same property.
\end{proof}

\subsection{Virtual tangent bundle}\label{sec:5.2}

Given $\xi\neq 0$, the tangent bundle of ${\cM}^{1,1,1}(M_1,M_2,M_3)$ is the cohomology of the complex
\ie\label{tangent of our moduli}
\bigoplus^3_{i=1}\mathrm{End}(\mathcal V_{i})\overset{\sigma}{\longrightarrow} \bigoplus^3_{i=1}\mathrm{End}(\mathcal V_{i})^{\oplus 2}\bigoplus_{i=1}^3 \mathcal V_{i} \bigoplus_{i=1}^3 \mathcal V_{i}^{*}\bigoplus_{1\leq i\neq j}^{3} \mathrm{Hom}(\mathcal V_{i},\mathcal V_{j})
\overset{d\mu_{\bC}}{\longrightarrow} \bigoplus^3_{i=1}\mathrm{End}(\mathcal V_{i})
\fe
where the middle term has cohomology degree zero. $d\mu_{\bC}$ is the differential of the moment map and is explicitly written as
\ie 
(x_i,y_i,i_i,j_i,s_{ij},s_{ki})\mapsto [X_i,y_i]+[x_i,Y_i]+i_iJ_i+I_ij_i+s_{ij}T_{ji}+s_{ik}T_{ki}.
\fe
$\sigma$ descends from the action of $\mathfrak{gl}_{M_1}\times\mathfrak{gl}_{M_2}\times\mathfrak{gl}_{M_3}$, and explicitly it is written as 
\ie 
(\Lambda_i,\Lambda_j)\mapsto([\Lambda_i,X_i],[Y_i,\Lambda_i],\Lambda_iI_i,-J_i\Lambda_i,\Lambda_jS_{ij}-S_{ij}\Lambda_i,\Lambda_iT_{ji}-T_{ji}\Lambda_j).
\fe
It is easy to verify that $d\mu_{\bC}\circ \sigma=0$. Note that $\sigma$ is injective and $d\mu_{\bC}$ is surjective by stability. Our moduli space $\cM^{1,1,1}(M_1,M_2,M_3)$ embeds into $\widetilde{\cM}^{1,1,1}(M_1,M_2,M_3)$ regularly, and is locally defined by $M_1M_2+M_2M_3+M_3M_1$ equations coming from matrix elements of $\mathcal T\in \mathrm{Hom}(\mathcal V_{j},\mathcal V_{i})$, so $\cM^{1,1,1}(M_1,M_2,M_3)$ has \textit{perfect obstruction theory} 
\ie
T\widetilde{\cM}^{1,1,1}(M_1,M_2,M_3)|_{\cM^{1,1,1}(M_1,M_2,M_3)}\to \mathrm{Hom}(\mathcal V_{2},\mathcal V_{1})\oplus \mathrm{Hom}(\mathcal V_{1},\mathcal V_{3})\oplus \mathrm{Hom}(\mathcal V_{3},\mathcal V_{2})
\fe
And this complex is quasi-isomorphic to the tangent complex $\mathbb{T}\cM^{1,1,1}(M_1,M_2,M_3)$. This obstruction theory can be written as a complex of tautological bundles:
\ie\label{tangent of our moduli2}
\bigoplus^3_{i=1}\mathrm{End}(\mathcal V_{i})\overset{\sigma}{\longrightarrow} \bigoplus^3_{i=1}\mathrm{End}(\mathcal V_{i})^{\oplus 2}\bigoplus_{i=1}^3 \mathcal V_{i} \bigoplus_{i=1}^3 \mathcal V_{i}^{*}\bigoplus_{1\leq i\leq j}^{3} \mathrm{Hom}(\mathcal V_{i},\mathcal V_{j})
\overset{d\mu_{\bC}}{\longrightarrow} \bigoplus^3_{i=1}\mathrm{End}(\mathcal V_{i})
\fe
The difference between \eqref{tangent of our moduli} and \eqref{tangent of our moduli2} is that $\mathrm{Hom}(\mathcal V_{i},\mathcal V_{j})$ with $i>j$ drop out and chain maps are restricted to the rest of components. In \eqref{tangent of our moduli}, $\sigma$ is injective by the stability condition, but $d\mu_{\bC}$ is not necessarily surjective. We shall show that $d\mu_{\bC}$ is not surjective at some points in $\cM^{1,1,1}(M_1,M_2,M_3)$.
\subsection{Torus fixed points}\label{sec:5.4}

Before starting the discussion of torus fixed points, let us introduce an auxiliary quiver. Consider the ``semi-infinite'' quiver with gauge nodes labeled with natural numbers $n\in \mathbb N$, and an arrow $S_n$ from the $n$'th node to the $n+1$'st node, and two self-loops $X_n,Y_n$ at the node $n$, and finally a framing at the zeroth node with arrows $I,J$. We impose moment map equations on these arrows:
\ie 
\left[X_n,Y_n\right]=0&,\; n>0;\\
\left[X_n,Y_n\right]+IJ=0&,\; n=0.
\fe
Suppose that the rank of the $n$' th gauge node is $N[n]$ and that $N[n]=0$ if $n\gg 0$, then denote by $\cM^K(N[\star])$ the quiver variety associated to the above quiver with gauge nodes ranks $N[0],N[1],\cdots$ and framing rank $K$, together with the stability conditions 
\ie 
\bC\langle X_n,Y_n\rangle\mathrm{Im}(S_{n-1})&=\bC^{N[n]},\; n>0\\
\bC\langle X_n,Y_n\rangle\mathrm{Im}(I)&=\bC^{N[n]},\; n=0.
\fe
\begin{figure}[H]
    \centering
    \vspace{-0.6cm}
    \includegraphics[width=10cm]{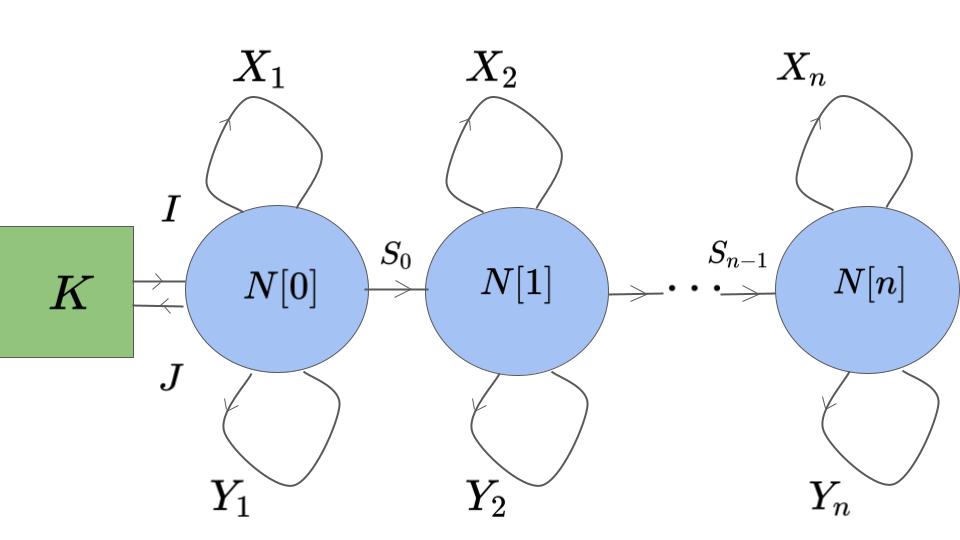}
    \caption{Auxiliary quiver}
    \vspace{-0.5cm}
    \label{fig:doublecontraction}
     \centering
\end{figure}
This is the prolongation of the quiver that appeared in \cite{DelZotto:2021ydd}. The following is obvious.
\begin{lemma}
Given a sequence of non-negative integers $N[\star]=\{N[0],\cdots,N[n-1],0,0\cdots\}$, and let $N[n]$ be a positive integer and form another sequence of non-negative integers as $\widetilde N[\star]=\{N[0],\cdots,N[n-1],N[n],0,0\cdots\}$, then $\cM^K(\widetilde N[\star])$ is a locally trivial fibration over $\cM^K(N[\star])$ with fibers isomorphic to $\mathrm{Quot}^{N[n]}(\cO^{N[n-1]}_{\bC^2})$.
\end{lemma}
This lemma tells us that $\cM^K(N[\star])$ is built from inductive fibrations with fibers isomorphic to $\mathrm{Quot}^{N[n]}(\cO^{N[n-1]}_{\bC^2})$ and the starting base is $\cM^K(\{N[0],0,0\cdots\})$, which is isomorphic to the ADHM moduli space $\cM(N[0],K)$.

Returning to our quiver in Figure \ref{fig:doublecontraction}, we introduce some notation of the tori acting on it. Define the action of tori $\mathbf{T}_{i},i\in \{1,2,3\}$, $\tilde{\mathbf{T}}_i,i\in\{1,2,3\}$ as follows:
\ie 
a_i\in \mathbf{T}_{i}&:(X_i,Y_i)\mapsto (a_iX_i,a_i^{-1}Y_i),\\
b_i\in\tilde{\mathbf{T}}_{i}&:S_{ij}\mapsto b_iS_{ij},
\fe
and denote by $\mathbf{T}_I,\tilde{\mathbf{T}}_I,I\in\{1,2,3\}$ the torus $\prod_{i\in I}\mathbf{T}_i\times\tilde{\mathbf{T}}_i$. 

The next proposition follows from the same argument as the proof of \cite[Proposition 2.3.1]{Maulik:2012}
\begin{proposition}
Assume that $\xi>0$, then the $\prod_{i\in I}\tilde{\mathbf{T}}_i$ fixed points have a disjoint union decomposition:
\ie 
\cM^{1,1,1}(M_1,M_2,M_3)^{\prod_{i\in I}\tilde{\mathbf{T}}_i}=\bigsqcup_{\substack{\sum_{n\ge 0} M_1[3n]+M_2[3n+2]+M_3[3n+1]=M_1\\ \sum_{n\ge 0} M_2[3n]+M_3[3n+2]+M_1[3n+1]=M_2\\ \sum_{n\ge 0} M_3[3n]+M_1[3n+2]+M_2[3n+1]=M_3}}\cM^{1}(M_1[\star])\times\cM^{1}(M_2[\star])\times \cM^{1}(M_3[\star]).
\fe
Here $M_i[\star],i=1,2,3$ are sequences of nonnegative integers, and $\cM^{1}(M_i[\star]),i=1,2,3$ are quiver varieties associated with the auxiliary quiver introduced in the beginning of this subsection.
\end{proposition} 

More generally, we have a nine dimensional torus $\prod^3_{i=1}\mathbf{T}'_i\times\prod^3_{i=1} \tilde{\mathbf{T}}_i$ acting on $\cM^{1,1,1}(M_1,M_2,M_3)$, which extends the action of $\mathbf{T}_{I}\times\tilde{\mathbf{T}}_I$. The modified action is defined by 
\ie 
(a^{(1)}_i,a^{(2)}_i)\in \mathbf{T}'_{i}&:(X_i,Y_i,I_i,J_i)\mapsto (a^{(1)}_iX,a^{(2)}_iY_i,I_i,a^{(1)}_ia^{(2)}_iJ_i)
\fe
Denote this torus by $\mathbf{T}_{\mathrm{edge}}$, then we have the following:
\begin{proposition}\label{counting fixed pts}
Under the above assumptions, $\cM^{1,1,1}(M_1,M_2,M_3)^{\mathbf{T}_{\mathrm{edge}}}$ is disjoint union of points \footnote{In the scheme-theoretical sense, i.e. there is no infinitesimal deformation of those points inside the $\mathbf{T}_{\mathrm{edge}}$-fixed points loci.}, points one-to-one correspond to $1+\prod_{j=0}^{n-1}M_1[3j]M_2[3j+1]M_3[3j+2]+1+\prod_{j=0}^{n-1}M_2[3j]M_3[3j+1]M_1[3j+2]+1+\prod_{j=0}^{n-1}M_3[3j]M_1[3j+1]M_2[3j+2]$ copies of Young diagrams. Moreover, the generating function for the number of $\mathbf{T}_{\mathrm{edge}}$-fixed points can be written as
\ie\label{fixed pts generating function}
\sum_{M_1,M_2,M_3}\#\left(\cM^{1,1}(M_1,M_2,M_3)^{\mathbf{T}_{\mathrm{edge}}}\right)x^{M_1}y^{M_2}z^{M_3}=\eqref{combinatorics},
\fe
\end{proposition}

\begin{proof}
Recall that the edge torus fixed points of the ADHM quiver variety $\cM^{1}(M_1[0])$ is disjoint union of points that are one-to-one correspond to Young diagrams $Y^{(1)}_{0,0}$ with $|Y^{(1)}_{0,0}|=M_1[0]$. More precisely, a fixed point $p\in \cM^{1}(M_1[0])^{\mathbf{T}'_1}$ corresponds to a unique $\mathbf{T}'_1$-equivariant quiver representation of which $\mathbb C^{M_1[0]}$ decomposes into $\mathbf{T}'_1$ weight spaces
$$\mathbb C^{M_1[0]}=\bigoplus_{(i,j)\in Y^{(1)}_{0,0}} \mathbb C\cdot e_{i,j},$$
where $(i,j)\in Y^{(1)}_{0,0}$ is the box $(i,j)$ in the Young diagram $Y^{(1)}_{0,0}$ and $e_{i,j}$ is a vector of weight $(-i,-j)$ under the $\mathbf{T}'_1$ action. We set $e_{i,j}=0$ if $(i,j)\notin Y^{(1)}_{0,0}$, then $X$ maps $e_{i,j}$ to $e_{i+1,j}$ and $Y$ maps $e_{i,j}$ to $e_{i,j+1}$ and $I$ maps the flavor vector to $e_{0,0}$ and all other arrows are zero.

Returning to our situation, we have 
\ie
\cM^{1,0,0}(M_1[\star])\times \cM^{0,1,0}(M_2[\star])\times \cM^{0,0,1}(M_3[\star]).
\fe
We will first focus on $\cM^{1,0,0}(M_1[\star])$ as
other factors $\cM^{0,1,0}(M_2[\star])$ and $\cM^{0,0,1}(M_3[\star])$ can be easily obtained by applying a proper $S_3$ action on the result of $\cM^{1,0,0}(M_1[\star])$. As it has an 3n-node tail, we would like to start from $n=1$ example. 

First, notice that there is a map 
\ie
\pi_1:\cM^{1,0,0}(M_1[0],M_2[1],M_3[2])^{\mathbf{T}'_1\times\mathbf{T}'_2\times\mathbf{T}'_3}\to \cM^{1,0}(M_1[0],M_2[1])^{\mathbf{T}'_1\times\mathbf{T}'_2}
\fe
and the latter is a disjoint union of points labeled by a Young diagram. Fix Young diagrams $Y_{k,0}^{(2)}$ so that $\sum_{k=1}^{M_1[0]}|Y_{k,0}^{(2)}|=M_2[1]$. Then $Y_{k,0}^{(2)}$ corresponds to a point $p$ in $\cM^{1,0,0}(M_1[0],M_2[1],0)^{\mathbf{T}'_1\times \mathbf{T}'_2}$, and a point in the preimage of $p$ in $\cM^{1,0,0}(M_1[0],M_2[1],M_3[2])^{\mathbf{T}'_1\times\mathbf{T}'_2\times\mathbf{T}'_3}$ corresponds to a $\mathbf{T}'_1\times\mathbf{T}'_2\times\mathbf{T}'_3$-equivariant quiver representation of which $\mathbb C^{M_2[1]}$ decomposes into $\mathbf{T}'_1\times \mathbf{T}'_2\times \mathbf{T}'_3$ weight spaces
$$\mathbb C^{M_2[1]}=\bigoplus_{(i,j)\in Y_{k,0}^{(2)}} \mathbb C\cdot e_{i,j}.$$
Here $(i,j)\in Y_{k,0}^{(2)}$ is the $(i,j)$-th box in the Young diagram $Y_{k,0}^{(2)}$ and $e_{i,j}$ is a vector of weight $(-i,-j)$ under the $\mathbf{T}'_2$ action and of weight zero under the $\mathbf{T}'_1$, $\mathbf{T}'_3$ action.

Next, we decompose $\mathbb C^{M_3[2]}$ into $\mathbf{T}'_2$ weight spaces:
$$\mathbb C^{M_3[2]}=\bigoplus_{(i,j)\in Y_{k,0}^{(2)}} V_{i,j},$$
where $\mathbf{T}'_2$ weight $(-i,-j)$. By equivariance, $S_{23}$ map $e_{i,j}$ to $V_{i,j}$ and $X_3,Y_3$ map $V_{i,j}$ to itself. Furthermore, by stability, we have 
\ie
\mathbb C\la X_3, Y_3\ra (S_{23})(e_{i,j})= V_{i,j}.
\fe
This shows that the restriction of $(X_3, Y_3)$ to $V_{i,j}$ is stable ADHM quiver representations respectively, where the flavour vector is $e_{i,j}$. In other words, we show that the preimage of $p$ in $\cM^{1,0,0}(M_1[0],M_2[1],M_3[2])^{\mathbf{T}'_1\times\mathbf{T}'_2\times\mathbf{T}'_3}$ is isomorphic to
\ie\label{1stdecomp}
\bigsqcup_{\substack{N_1,\cdots,N_{M_2[1]}\ge 0\\N_1+\cdots +N_{M_2[1]}=M_3[2]}}\prod_{i=1}^{M_2[1]}
\cM^{0,0,1}(0,0,N_i)^{\mathbf{T}'_3}.
\fe
Each of the factors $\cM^{0,0,1}(0,0,N_i)^{\mathbf{T}'_3}$ one-to-one corresponds to Young diagrams $Y^{(3)}_{i,0}$ with $N_i$ boxes. Therefore, we prove that $\cM^{1,1,0}(M_1[0],M_2[1],M_3[2])^{\mathbf{T}_{\mathrm{edge}}}$ is a disjoint union of points that correspond one-to-one to $M_2[1]$ copies of the Young diagrams $Y^{(3)}_{1,0},\cdots, Y^{(3)}_{M_2[1],0}$, where $\sum_{i=1}^{M_2[1]}|Y^{(3)}_{i,0}|=M_3[2]$.

Next, consider the second map \ie
\pi_2:\cM^{1,0,0}(M_1[0],M_2[1])^{\mathbf{T}'_1\times\mathbf{T}'_2}\to \cM^{1}(M_1[0])^{\mathbf{T}'_1}.\fe Let $p'$ be a point in $\cM^{1}(M_1[0])^{\mathbf{T}'_1}$. Similarly to the above, one can argue that $\pi_2^{-1}(p')$ is  
\ie\label{2nddecomp}
\bigsqcup_{\substack{N_1,\cdots,N_{M_1[0]}\ge 0\\N_1+\cdots +N_{M_1[0]}=M_2[1]}}\prod_{i=1}^{M_1[0]}
\cM^{0,1,0}(0,N_i,0)^{\mathbf{T}'_2}.
\fe

We can then combine \eqref{1stdecomp} and \eqref{2nddecomp} and conclude that $\cM^{1,0,0}(M_1[0],M_2[1],M_3[2])^{\bf{T}_{\text{edge}}}$ is disjoint union of points which are in one-to-one correspondence with $1+M_1[0]M_2[1]$ copies of Young diagrams 
\ie
Y^{(1)}_{0,0},~Y^{(2)}_{1,0}\left[Y^{(3)}_{1,0}, \cdots, Y^{(3)}_{M_2[1],0}\right], \cdots, Y^{(2)}_{M_1[0],0}\left[Y^{(3)}_{1,0}, \cdots, Y^{(3)}_{M_2[1],0}\right],\fe
where 
\ie
\sum^{M_1[0]}_{i=1}|Y^{(2)}_{i,0}|=M_2[1],\quad \sum^{M_2[1]}_{i=1}|Y^{(3)}_{i,0}|=M_3[2].
\fe
Here, the notation $Y\left[Y'_1\ldots Y'_n\right]$ indicates that we assign $Y'_i$ for each box $i$ of Y when $|Y|=n$. Now that we have finished analyzing the basic 3-node unit of the 3n node quiver $\cM^{1,0,0}(M_1[\star])$. 

The general pattern is that each of the Young diagrams in i-th node of the quiver corresponds to one of the boxes in one of the Young diagrams in $(i-1)-$th node. With this intuition, one can proceed to the $3n$-th node and conclude that there are 
\ie\nonumber
1+M_1[0]M_2[1]M_3[2]\cdots M_1[3n-3]M_2[3n-2]
\fe
copy of Young diagrams, where
\ie
\sum^{M_1[3j-3]}_{i=1}|Y^{(2)}_{i,j}|=M_2[3j-2],\quad\sum^{M_2[3j-2]}_{i=1}|Y^{(3)}_{i,j}|=M_3[3j-1],\quad\sum^{M_3[3j-1]}_{i=1}|Y^{(1)}_{i,j}|=M_1[3j].
\fe

Similarly, the space $\cM^{0,1,0}(M_2[\star])$ consists of a set of fixed points represented by 
\ie\nonumber
1+M_2[0]M_3[1]M_1[2]\cdots M_2[3n-3]M_3[3n-2]
\fe
copy of Young diagrams, where
\ie
\sum^{M_2[3j-3]}_{i=1}|Y^{(3)}_{i,j}|=M_3[3j-2],\quad\sum^{M_3[3j-2]}_{i=1}|Y^{(1)}_{i,j}|M_1[3j-1],\quad\sum^{M_1[3j-1]}_{i=1}|Y^{(2)}_{i,j}|=M_2[3j].
\fe
Lastly, the space $\cM^{0,0,1}(M_3[\star])$ consists of a set of fixed points represented by 
\ie\nonumber
1+M_3[0]M_1[1]M_2[2]\cdots M_3[3n-3]M_1[3n-2]
\fe
copy of Young diagrams, where
\ie
\sum^{M_3[3j-3]}_{i=1}|Y^{(1)}_{i,j}|=M_1[3j-2],\quad\sum^{M_1[3j-2]}_{i=1}|Y^{(2)}_{i,j}|=M_2[3j-1],\quad\sum^{M_2[3j-1]}_{i=1}|Y^{(3)}_{i,j}|=M_3[3j].
\fe

We now want to enumerate all fixed points. To do so, first notice that 
\ie
    \sum_{M_1} \left(\cM^{1,0,0}(M_1,0,0)^{\mathbf{T}_{\mathrm{edge}}}\right)x^{M_1}=\sum_{M_1}p(M_1)x^{M_1}=\prod_{i\ge 1}^{\infty}\frac{1}{1-x^i}=\frac{1}{(x;x)},
\fe
where $p(n)$ is Euler's partition function counting the number of partitions of $n$ in positive integers and $\left(a;y\right)_\infty=\prod^\infty_{k=0}(1-ay^k)$ is the Pochhammer symbol.

Using the above information, we can write down a concise form of the generating function of the fixed points labeled by the Young diagrams.
\ie\nonumber
&\sum_{M_1,M_2,M_3}\left(\cM^{1,0,0}(M_1[\star])\times \cM^{0,1,0}(M_2[\star])\times \cM^{0,0,1}(M_3[\star])\right)x^{M_1}y^{M_2}z^{M_3}\\
=&\sum_{M_1[0]}\left(\cM^{1,0,0}(M_1[0],0,0)^{\mathbf{T}'_{1}}\right)\sum_{N_1,\cdots,N_{M_1[0]}\ge 0}\prod_{i=1}^{M_1[0]}\left(\cM^{0,1,0}(0,N_{i},0)^{\mathbf{T}'_2}\right)x^{M_1[0]}y^{N_1+\cdots+N_{M_1[0]}}\\
\times&\sum_{M_2[1]}\left(\cM^{0,1,0}(0,M_2[1],0)^{\mathbf{T}'_{2}}\right)\sum_{N_1,\cdots,N_{M_2[1]}\ge 0}\prod_{i=1}^{M_2[1]}\left(\cM^{0,0,1}
(0,0,N_i)^{\bf{T}'_3}\right)y^{M_2[1]}z^{N_1+\cdots+N_{M_2[1]}}\times\ldots\\
\times&\sum_{M_2[3j-2]}\left(\cM^{0,1,0}(0,M_2[3j-2],0)^{\mathbf{T}'_{2}}\right)\sum_{N_1,\cdots,N_{M_2[3j-2]}\ge 0}\prod_{i=1}^{M_2[3j-2]}\left(\cM^{0,0,1}
(0,0,N_i)^{\bf{T}'_3}\right)y^{M_2[3j-2]}z^{N_1+\cdots+N_{M_2[3j-2]}}\\
\\
\times&\sum_{M_2[0]}\left(\cM^{0,1,0}(0,M_2[0],0)^{\mathbf{T}'_{2}}\right)\sum_{N_1,\cdots,N_{M_2[0]}\ge 0}\prod_{i=1}^{M_2[0]}\left(\cM^{0,0,1}(0,0,N_{i})^{\mathbf{T}'_3}\right)y^{M_2[0]}z^{N_1+\cdots+N_{M_2[0]}}\\
\times&\sum_{M_3[1]}\left(\cM^{0,0,1}(0,0,M_3[1])^{\mathbf{T}'_{3}}\right)\sum_{N_1,\cdots,N_{M_3[1]}\ge 0}\prod_{i=1}^{M_3[1]}\left(\cM^{1,0,0}
(N_i,0,0)^{\bf{T}'_1}\right)z^{M_3[1]}x^{N_1+\cdots+N_{M_3[1]}}\times\ldots\\
\times&\sum_{M_3[3j-2]}\left(\cM^{0,0,1}(0,0,M_3[3j-2])^{\mathbf{T}'_{3}}\right)\sum_{N_1,\cdots,N_{M_3[3j-2]}\ge 0}\prod_{i=1}^{M_3[3j-2]}\left(\cM^{1,0,0}
(N_i,0,0)^{\bf{T}'_1}\right)z^{M_3[3j-2]}x^{N_1+\cdots+N_{M_3[3j-2]}}\\
\times&\sum_{M_3[0]}\left(\cM^{0,0,1}(0,0,M_3[0])^{\mathbf{T}'_{3}}\right)\sum_{N_1,\cdots,N_{M_3[0]}\ge 0}\prod_{i=1}^{M_3[0]}\left(\cM^{1,0,0}(N_{i},0,0)^{\mathbf{T}'_1}\right)z^{M_3[0]}x^{N_1+\cdots+N_{M_3[0]}}\\
\times&\sum_{M_1[1]}\left(\cM^{1,0,0}(M_1[1],0,0)^{\mathbf{T}'_{1}}\right)\sum_{N_1,\cdots,N_{M_1[1]}\ge 0}\prod_{i=1}^{M_1[1]}\left(\cM^{1,0,0}
(0,N_i,0)^{\bf{T}'_2}\right)x^{M_1[1]}y^{N_1+\cdots+N_{M_1[1]}}\times\ldots\\
\times&\sum_{M_1[3j-2]}\left(\cM^{1,0,0}(M_1[3j-2],0,0)^{\mathbf{T}'_{1}}\right)\sum_{N_1,\cdots,N_{M_1[3j-2]}\ge 0}\prod_{i=1}^{M_1[3j-2]}\left(\cM^{1,0,0}
(0,N_i,0)^{\bf{T}'_2}\right)x^{M_1[3j-2]}y^{N_1+\cdots+N_{M_1[3j-2]}}
\fe
\ie\label{combinatorics}
=&\sum_{M_1[0]}p(M_1[0])\left(\frac{x}{(y;y)}\right)^{M_1[0]}
\sum_{M_2[1]}p(M_2[1])\left(\frac{y}{(z;z)}\right)^{M_2[1]}\cdots\sum_{M_2[3n-2]}p(M_2[3n-2])\left(\frac{y}{(z;z)}\right)^{M_2[3n-2]}\\
\times&\sum_{M_2[0]}p(M_2[0])\left(\frac{y}{(z;z)}\right)^{M_2[0]}
\sum_{M_3[1]}p(M_3[1])\left(\frac{z}{(x;x)}\right)^{M_3[1]}\cdots\sum_{M_3[3n-2]}p(M_3[3n-2])\left(\frac{z}{(x;x)}\right)^{M_3[3n-2]}\\
\times&\sum_{M_3[0]}p(M_3[0])\left(\frac{z}{(x;x)}\right)^{M_3[0]}
\sum_{M_1[1]}p(M_1[1])\left(\frac{x}{(y;y)}\right)^{M_1[1]}\cdots\sum_{M_1[3n-2]}p(M_1[3n-2])\left(\frac{x}{(y;y)}\right)^{M_1[3n-2]}.
\fe
\end{proof}
\subsection{Equivariant K-theory index}\label{sec:5.5}
The equivariant K-theory index of the moduli space $\cM^{1,1,1}(M_1,M_2,M_3)$ is defined by
\ie 
\chi(\cM^{1,1,1}(M_1,M_2,M_3),\mathcal O^{\mathrm{vir}})
\fe
where $\mathcal O^{\mathrm{vir}}$ is the virtual structure sheaf of $\cM^{1,1,1}(M_1,M_2,M_3)$. Since our moduli space $\cM^{1,1,1}(M_1,M_2,M_2)$ is l.c.i by Proposition \ref{prop: regular embedding}, the virtual structure sheaf is the actual structure sheaf, i.e. $\mathcal O^{\mathrm{vir}}=\mathcal O$. The equivariant K-theory partition function is defined by
\ie
Z_{1,1,1}&=\sum_{M_1,M_2,M_3}\chi(\cM^{1,1,1}(M_1,M_2,M_3),\mathcal O^{\mathrm{vir}})x^{M_1}y^{M_2}z^{M_3}\\
&=\sum_{M_1,M_2,M_3}\chi(\cM^{1,1,1}(M_1,M_2,M_3))x^{M_1}y^{M_2}z^{M_3}
\fe
Using equivariant localization, the $\mathbf{T}_{\mathrm{edge}}$-equivariant K-theory index of $\cM^{1,1,1}(M_1,M_2,M_3)$ can be written in terms of $\mathbf{T}_{\mathrm{edge}}$-fixed points:
\ie 
\chi(\cM^{1,1,1}(M_1,M_2,M_3),\mathcal O)=\chi\left(\cM^{1,1,1}(M_1,M_2,M_3)^{\mathbf{T}_{\mathrm{edge}}},\left(S^{\bullet}(T^{\mathrm{vir}})^*\right)|_{\cM^{1,1,1}(M_1,M_2,M_3)^{\mathbf{T}_{\mathrm{edge}}}}\right)
\fe
where $T^{\mathrm{vir}}$ is the virtual tangent bundle \eqref{tangent of our moduli}. 

Let us denote the equivariant parameters of $\mathbf{T}'_i$ by $r_{1,i},r_{2,i}$, and the equivariant parameters of $\tilde{\mathbf{T}}'_i$ by $s_{i}$. Then the $\mathbf{T}_{\mathrm{edge}}$-equivariant K-theory class of the virtual tangent bundle \eqref{tangent of our moduli} can be written as 
\ie\label{virtual tangent character}
\sum_{i=1}^3\cV_i+\cV_i^*r_{1,i}r_{2,i}-\cV_i\cV_i^*(1-r_{1,i})(1-r_{2,i})+s_{i}\cV_i\cV_{i+1}^*
\fe
where $\cV_i$ are universal bundles on $\cM^{1,1,1}(M_1,M_2,M_3)$ of rank $M_i$ and whenever the index is greater than 3, we take mod 3. Note that they are equivariant under the $\mathbf{T}_{\mathrm{edge}}$ action. At a fixed point labeled by
\ie
Y^{(1)}_{a,b},~Y^{(2)}_{a,b},~Y^{(3)}_{a,b},
\fe
where $b=0,\ldots,n-1$ and $a$ depend on $b$ and the superscript $(c)$ of $Y^{(c)}_{a,b}$.

The fibers of $\cV_i$ are representations of $\mathbf{T}_{\mathrm{edge}}$, and their $\mathbf{T}_{\mathrm{edge}}$ characters can be described as follows. Let us introduce the notation
\ie 
w\left(Y^{(1)}_{a,b}\right)=\sum_{(i,j)\in Y^{(1)}_{a,b}}r_{1,1}^{-i}r_{2,1}^{-j},\quad w\left(Y^{(2)}_{a,b}\right)=\sum_{(i,j)\in Y^{(2)}_{a,b}}r_{1,2}^{-i}r_{2,2}^{-j},\quad w\left(Y^{(3)}_{a,b}\right)=\sum_{(i,j)\in Y^{(3)}_{a,b}}r_{1,3}^{-i}r_{2,3}^{-j}.
\fe

We will focus on $\cM^{1}(M_1[\star])$ first. Using the description of the fixed points $\mathbf{T}_{\mathrm{edge}}$, we know that $Y^{(2)}_{a_2,b_2;i_1,j_1}$ is in one-to-one correspondence with $(i_1,j_1)$-th box in the Young diagram $Y^{(1)}_{a_1,b_1}$, where we used the notation $(a_2,b_2;i_1,j_1)$ to keep track of the framing of $Y^{(2)}_{a_2,b_2}$ at $(i_1,j_1)$-th box of $Y^{(1)}_{a_1,b_1}$. Similarly,  $Y^{(3)}_{a_3,b_3;i_2,j_2}$ is in one-to-one correspondence with a box $(i_2,j_2)$ in the diagram $Y^{(2)}_{a_2,b_2}$. Then the characters of $\cV^{(a_1,b_1)}_1$, $\cV_2^{(a_2,b_2)}$, $\cV^{(a_3,b_3)}_3$ are
\ie 
\cV_1&=w\left(Y^{(1)}_{a_1,b_1}\right)\text{ if }(a_1,b_1)=(0,0)\text{ else }s_{1}^{-1}\sum_{(i_3,j_3)\in Y^{(3)}_{a_3,b_3}}r_{1,3}^{-i_3}r_{2,3}^{-j_3}w\left(Y^{(1)}_{a_1,b_1;i_3,j_3}\right),\\ 
\cV_2&=s_2^{-1}\sum_{(i_1,j_1)\in Y^{(1)}_{a_1,b_1}}r_{1,1}^{-i_1}r_{2,1}^{-j_1}w\left(Y^{(2)}_{a_2,b_2;i_1,j_1}\right),~~\cV_3=s_{3}^{-1}\sum_{(i_2,j_2)\in Y^{(2)}_{a_2,b_2}}r_{1,2}^{-i_2}r_{2,2}^{-j_2}w\left(Y^{(3)}_{a_3,b_3;i_2,j_2}\right),
\fe
where $Y^{(1)}_{0,0}$ is the Young diagram that represents the fixed point of the vector space of dimension $M_1[0]$ corresponding to the first node of the elongated linear quiver $\cM^1(M_1[\star])$. It stands out from the others, since it is directly connected to the framing node.

Similarly, for $\cM^{1}(M_2[\star])$, characters of $\cV'_1$, $\cV'_2$, $\cV'_3$ are
\ie
\cV'_2&=w\left(Y^{(2)}_{a_2,b_2}\right)\text{ if }(a_2,b_2)=(0,0)\text{ else }s_{2}^{-1}\sum_{(i_1,j_1)\in Y^{(1)}_{a_1,b_1}}r_{1,1}^{-i_1}r_{2,1}^{-j_1}w\left(Y^{(2)}_{a_2,b_2;i_1,j_1}\right),\\ 
\cV'_3&=s_{3}^{-1}\sum_{(i_2,j_2)\in Y^{(2)}_{a_2,b_2}}r_{1,2}^{-i_2}r_{2,2}^{-j_2}w\left(Y^{(3)}_{a_3,b_3;i_2,j_2}\right),~~\cV'_1=s_1^{-1}\sum_{(i_3,j_3)\in Y^{(3)}_{a_3,b_3}}r_{1,3}^{-i_3}r_{2,3}^{-j_3}w\left(Y^{(1)}_{a_1,b_1;i_3,j_3}\right),
\fe
where $Y^{(2)}_{0,0}$ represents the fixed point of the first node of $\cM^1(M_2[\star])$.

Lastly, for $\cM^{1}(M_3[\star])$, characters of $\cV''_1$, $\cV''_2$, $\cV''_3$ are
\ie
\cV''_3&=w\left(Y^{(3)}_{a_3,b_3}\right)\text{ if }(a_3,b_3)=(0,0)\text{ else }s_{3}^{-1}\sum_{(i_2,j_2)\in Y^{(2)}_{a_2,b_2}}r_{1,2}^{-i_2}r_{2,2}^{-j_2}w\left(Y^{(3)}_{a_3,b_3;i_2,j_2}\right),\\ 
\cV''_1&=s_{1}^{-1}\sum_{(i_3,j_3)\in Y^{(3)}_{a_3,b_3}}r_{1,3}^{-i_3}r_{2,3}^{-j_3}w\left(Y^{(1)}_{a_1,b_1;i_3,j_3}\right),~~\cV''_2=s_2^{-1}\sum_{(i_1,j_1)\in Y^{(1)}_{a_1,b_1}}r_{1,1}^{-i_1}r_{2,1}^{-j_1}w\left(Y^{(2)}_{a_2,b_2;i_1,j_1}\right),\fe
where $Y^{(3)}_{0,0}$ represents the fixed point of the first node of $\cM^1(M_3[\star])$.

For example, if $\Phi=\mathcal O$, then the $\mathbf{T}_{\mathrm{edge}}$-equivariant K-theory index consists of the following:
\ie\label{Index for Phi=0}
\chi(\cM^{1}(M_1[\star]),\mathcal O)&=\sum_{Y^{(1)}_{0,0},Y^{(2)}_{1,0},Y^{(3)}_{2,0},\ldots}S^{\bullet}(T^{\mathrm{vir}}(\cV_1,\cV_2,\cV_3))^*\\
\chi(\cM^{1}(M_2[\star])),\mathcal O)&=\sum_{Y^{(2)}_{0,0},Y^{(3)}_{1,0},Y^{(1)}_{2,0},\ldots}S^{\bullet}(T^{\mathrm{vir}}(\cV'_1,\cV'_2,\cV'_3))^*\\
\chi(\cM^{1}(M_3[\star])),\mathcal O)&=\sum_{Y^{(3)}_{0,0},Y^{(1)}_{1,0},Y^{(2)}_{2,0},\ldots}S^{\bullet}(T^{\mathrm{vir}}(\cV''_1,\cV''_2,\cV''_3))^*
\fe
where $T^{\mathrm{vir}}$ is given by \eqref{virtual tangent character}, and $S^{\bullet}(T^{\mathrm{vir}})$ is the plethystic exponential of the virtual character $T^{\mathrm{vir}}$, and the $*$ operator maps $(r_{1,i},r_{2,i},s_{i})$ to $(r_{1,i}^{-1},r_{2,i}^{-1},s_{i}^{-1})$.
\\
\newline
{\bf{Example}}\\
\newline
Let us compute the example $M_1=M_2=M_3=1$ explicitly. Recall that the moduli space $\cM^{1,1,1}(1,1,1)$ decomposes into $\cM^{1,0,0}(M_1[\star])\times\cM^{0,1,0}(M_2[\star])\times\cM^{0,0,1}(M_3[\star])$. Since $M_1=M_2=M_3=1$, we only have one torus fixed point
\ie
M_1[0]=1,\quad M_2[1]=0,\quad M_3[2]=0,\quad M_1[3j-3]=M_2[3j-2]=M_3[3j-1]=0,~\text{for }j\geq2\\
M_2[0]=1,\quad M_3[1]=0,\quad M_1[2]=0,\quad M_2[3j-3]=M_3[3j-2]=M_1[3j-1]=0,~\text{for }j\geq2,\\
M_3[0]=1,\quad M_1[1]=0,\quad M_2[2]=0,\quad M_3[3j-3]=M_1[3j-2]=M_2[3j-1]=0,~\text{for }j\geq2,
\fe
given by the following set of Young diagrams:
\ie
Y^{(1)}_{0,0}&=\square,~ Y^{(2)}_{1,0}=\emptyset,~
Y^{(3)}_{2,0}=\emptyset\; \text{and}\\
Y^{(2)}_{0,0}&=\square,~ Y^{(3)}_{1,0}=\emptyset,~
Y^{(1)}_{2,0}=\emptyset\;
\text{and}\\
Y^{(3)}_{0,0}&=\square,~ Y^{(1)}_{1,0}=\emptyset,~
Y^{(2)}_{2,0}=\emptyset.
\fe

The virtual tangent characters at these fixed points are 
\ie 
T_1^{\mathrm{vir}}(\cV_1,\cV_2,\cV_3)&=r_{1,1}+r_{2,1},\\
T_2^{\mathrm{vir}}(\cV'_1,\cV'_2,\cV'_3)&=r_{1,2}+r_{2,2}\\
T_3^{\mathrm{vir}}(\cV''_1,\cV''_2,\cV''_3)&=r_{1,3}+r_{2,3}.
\fe
If we plug these into \eqref{Index for Phi=0}, we get Euler characteristics of $\cM^{1,0,0}(1,1,1)$.
\ie
\chi(\cM^{1,1,1}(1,1,1),\mathcal O)&=P.E.\left[T_1^{\mathrm{vir}}(\cV_1,\cV_2,\cV_3)+T_2^{\mathrm{vir}}(\cV'_1,\cV'_2,\cV'_3)+T_3^{\mathrm{vir}}(\cV''_1,\cV''_2,\cV''_3)\right].
\fe
This example is quite trivial. The first non-trivial example can be worked out by starting from another fixed point:
\ie
M_1[0]=M_2[1]=M_3[2]=1,~\ldots\\
M_2[0]=M_3[0]=0,~\ldots
\fe
We would like to see the action of three copies of the 5d CS algebra on the partition function. To do so, we first observe there is a subvariety of the moduli space $\cM^{1,1,1}(M_1,M_2,M_3)$.\\
\newline
{\bf{An open subvariety of moduli space}}\\
\newline
We consider an open subvariety of $\cM^{1,1,1}(M_1,M_2,M_3)$ denoted by $\mathring{\cM}^{1,1,1}(M_1,M_2,M_3)$, which is defined by the open condition:
$$\bC\langle X_1,Y_1\rangle \mathrm{Im}(I_1)=\bC^{M_1},\quad \bC\langle X_2,Y_2\rangle \mathrm{Im}(I_2)=\bC^{M_2},\quad \bC\langle X_3,Y_3\rangle \mathrm{Im}(I_3)=\bC^{M_3}.$$
Under this condition, three gauge nodes satisfy their own stability condition when considered as individual ADHM quivers. Note that the conditions imply $J_i=0$; however, for rank 1 framing nodes, we can always restore the vanishing $J_i$ so that we recover the usual ADHM quivers. Hence, we have a natural projection given by
\ie\label{projto3}
p: \mathring{\cM}^{1,1,1}(M_1,M_2,M_3)\to \cM^1(M_1)\times \cM^1(M_2)\times\cM^1(M_3).
\fe

The projection map is a vector bundle with the fiber
\ie
\text{Hom}(\cV_1,\cV_2)\times\text{Hom}(\cV_2,\cV_3)\times\text{Hom}(\cV_3,\cV_1),
\fe
which are maps between $\mathcal V_{i}$ and $\mathcal V_{j}$, where $\mathcal V_{1}$ and $\mathcal V_{2}$ are universal bundles on $\cM^1(M_1)$, $\cM^1(M_2)$, $\cM^1(M_3)$ respectively. Now we can observe that the action of three copies of the 5d CS algebra on the open subvariety takes place by taking a particular limit of some of the equivariant parameters.
\begin{proposition}\label{prop: the limit}
In the limit $\underset{r_{1,1}r_{2,1}\to 1}{\lim}\:\underset{r_{1,2}r_{2,2}\to 1}{\lim}\:\underset{r_{1,3}r_{2,3}\to 1}{\lim}$, we have 
\ie\label{eqn: the limit}
\chi(\cM^{1,1,1}(M_1,M_2,M_3),\mathcal O)_{r_{1,i}r_{2,i}\to 1}=\chi(\mathring{\cM}^{1,1,1}(M_1,M_2,M_3),\mathcal O)_{r_{1,i}r_{2,i}\to 1},
\fe
in other words, the fixed points outside $\mathring{\cM}^{1,1,1}(M_1,M_2,M_3)$ do not contribute to this limit. 
\end{proposition}

\begin{remark}
The right-hand side of \eqref{eqn: the limit} does not depend on the choice of the order of taking the limit, but the left-hand side of \eqref{eqn: the limit} is indeed sensitive to the order of taking the limit, namely the order is important to ensure that the limit $\chi(\cM^{1,1,1}(M_1,M_2,M_3),\mathcal O)_{r_{1,i}r_{2,i}\to 1}$ exists. The choice of the order in the proposition is not the unique one that works, in fact it suffices to take one of following order $$\lim_{r_{1,1}r_{2,1}\to 1}\lim_{r_{1,2}r_{2,2}\to 1}\lim_{r_{1,3}r_{2,3}\to 1},\; \lim_{r_{1,3}r_{2,3}\to 1}\lim_{r_{1,1}r_{2,1}\to 1}\lim_{r_{1,2}r_{2,2}\to 1},\;\lim_{r_{1,2}r_{2,2}\to 1}\lim_{r_{1,3}r_{2,3}\to 1}\lim_{r_{1,1}r_{2,1}\to 1},$$ the limit exists and we get the same result which is the right hand side of \eqref{eqn: the limit}.
\end{remark}

\begin{proof}
Suppose that we have a fixed point labelled by a set of Young diagrams \ie
&Y^{(1)}_{0,0},~Y^{(2)}_{1,0;i_1,j_1},~Y^{(3)}_{2,0;i_2,j_2},~Y^{(1)}_{0,1;i_3,j_3},~\ldots,Y^{(1)}_{a_1,b_1;i_3,j_3},\ldots,Y^{(2)}_{a_2,b_2;i_1,j_1},\ldots,Y^{(3)}_{a_3,b_3;i_2,j_2},\ldots,\\
&Y^{(2)}_{0,0},~Y^{(3)}_{1,0;i_2,j_2},~Y^{(1)}_{2,0;i_3,j_3},~Y^{(2)}_{0,1;i_1,j_1},~\ldots,Y^{(2)}_{a_2,b_2;i_3,j_3},\ldots,Y^{(3)}_{a_3,b_3;i_2,j_2},\ldots,Y^{(1)}_{a_1,b_1;i_3,j_3},\\
&Y^{(3)}_{0,0},~Y^{(1)}_{1,0;i_3,j_3},~Y^{(2)}_{2,0;i_1,j_1},~Y^{(3)}_{0,1;i_2,j_2},~\ldots,Y^{(3)}_{a_3,b_3;i_2,j_2},\ldots,Y^{(1)}_{a_1,b_1;i_3,j_3},\ldots,Y^{(2)}_{a_2,b_2;i_1,j_1},\\
\fe
where $(i_k,j_k)$ takes value in the coordinate of box in the preceding Young diagram $Y^{(k)}$. We need to show that $S^{\bullet}(T^{\mathrm{vir}})$ vanishes in the limit $r_{1,1}r_{2,1},r_{1,2}r_{2,2},r_{1,3}r_{2,3}\to 1$ when there is $(i_1,j_1)\in Y^{(1)}_{0,0}$, $(i_2,j_2)\in Y^{(2)}_{0,0}$, $(i_3,j_3)\in Y^{(3)}_{0,0}$ respectively such that $Y^{(2)}_{1,0;i_1,j_1}$, $Y^{(3)}_{1,0;i_2,j_2}$, $Y^{(1)}_{1,0;i_3,j_3}$ are not empty. 

It is enough to show that the constant terms that appear in the virtual character $T^{\mathrm{vir}}_{r_{1,1}r_{2,1},r_{1,2}r_{2,2},r_{1,3}r_{2,3}\to 0}$ sourced from the Young diagrams $Y^{(2)}_{1,0;i_1,j_1}$, $Y^{(3)}_{1,0;i_2,j_2}$, $Y^{(1)}_{1,0;i_3,j_3}$ are negative. To understand why that is so, let us consider the following potential source for constant terms in $T^{\mathrm{vir}}$: $$\sum_{k_i}a_{k_i}(r_{1,i}r_{2,i})^{k_i}-\sum_{l_i}b_{l_i}(r_{1,i}r_{2,i})^{l_i},$$ where $a_{k_i},b_{l_i}$ are positive integers. After taking the plethystic exponential, this term becomes
$$\frac{\prod_{l_i}(1-r_{1,i}^{l_i}r_{2,i}^{l_i})^{b_{l_i}}}{\prod_{k_i}(1-r_{1,i}^{k_i}r_{2,i}^{k_i})^{a_{k_i}}}.$$
It becomes zero in the limit $r_{1,i}r_{2,i}\to 1$ if $\sum_{l_i} b_{l_i}>\sum_{k_i}a_{k_i}$. In other words, it suffices to prove that the constant terms in $T^{\mathrm{vir}}$ are negative and that they are derived from the Young diagrams, which do not exist in the fixed-point set of $\mathring{\cM}^{1,1,1}(M_1,M_2,M_3)$.

Let us find the constant term in $T^{\mathrm{vir}}_{r_{1,i}r_{2,i}\to 1}$. Let us introduce some notation. Recall
\ie 
\cV_2=\sum_{(i_1,j_1)\in Y^{(1)}_{a_1,b_1}}r_{1,1}^{-i_1}r_{2,1}^{-j_1}w\left(Y^{(2)}_{a_2,b_2;i_1,j_1}\right),\quad\cV_3=\sum_{(i_2,j_2)\in Y^{(2)}_{a_2,b_2}}r_{1,2}^{-i_2}r_{2,2}^{-j_2}w\left(Y^{(3)}_{a_3,b_3;i_2,j_2}\right).
\fe
For every $V_i\in K_{(\bC^\times)^2}(\mathrm{pt})=\bC[u_{1,i},u_{1,i}^{-1},u_{2,i},u_{2,i}^{-1}]$\footnote{Here, the notation $K_G(X)$ is an $G-$equivariant K-theory class of space $X$. When $X$ is a point, $K_G$(pt) is a representation space of $G$.}, define
\ie 
T(V_i)=V_i+V_i^*u_{1,i}u_{2,i}-V_iV_i^*(1-u_{1,i})(1-u_{2,i})
\fe
For example, let $(\bC^\times)^2$ be $\mathbf{T}'_1$ then $u_1=r_{1,1},u_2=r_{2,1}$, and $T(\cV_1)=\cV_1+\cV_1^*r_{1,1}r_{2,1}-\cV_1\cV_1^*(1-r_{1,i})(1-r_{2,i})$. Using this notation, we expand
\ie
T^{vir}=\sum_{i=1}^3\cV_i+\cV_i^*r_{1,i}r_{2,i}-\cV_i\cV_i^*(1-r_{1,i})(1-r_{2,i})+s_{i}\cV_i\cV_{i+1}^*
\fe
as follows
\ie\label{expand virtual tangent character}
T^{\mathrm{vir}}(\cV_1,\cV_2,\cV_3)=&T(\cV^{(0)}_1)+s_1^{-1}T(\cV_1)+s_2^{-1}T(\cV_2)+s_3^{-1}T(\cV_3)+\sum_{i=1}^3r_{1,i}r_{2,i}s_i\cV_{i+1}\\
&-\sum_{i=1}^3\cV_i\cV_i^{*}(1-r_{1,i})(1-r_{2,i})+\cV_{i+1}\cV_i^*.
\fe
Recall that if $V$ is the weight space of a Young diagram $Y$, then 
\ie\label{simple tangent character}
T(V)=\sum_{\square\in Y}u_1^{-l(\square)}u_2^{a(\square)+1}+\sum_{\square\in Y}u_1^{l(\square)+1}u_2^{-a(\square)},
\fe
where $a(\square)$ and $l(\square)$ are arm-length and leg-length respectively:
\ie 
a(\square)=\#\{j'>j|(i,j')\in Y\},\; l(\square)=\#\{i'>i|(i',j)\in Y\}.
\fe 
In particular, we see that there are no constant terms in $T(\cV_i)_{r_{1,2}r_{2,2},r_{1,3}r_{2,3}\to 1}$ for all $i=1,2,3$. 

Next, we simply drop all the other terms in \eqref{expand virtual tangent character} that involve $s_i$ or $s_i^{-1}$ because they cannot be constant in the limit $r_{1,i}r_{2,i}\to 0$. Therefore, we only need to calculate the constant term in 
\ie\label{reduced virtual tangent character}
-\sum_{i=1}^3\cV_i\cV_i^{*}(1-r_{1,i})(1-r_{2,i})+\cV_{i+1}\cV_i^*
\fe
in the limit $r_{1,i}r_{2,i}\to 1$. 

Let us take $\mathbf{T}'_1$-invariant first, this amounts to take those terms in \eqref{reduced virtual tangent character} which do not involve $r_{1,1}$ or $r_{2,1}$, and the result is
\ie
&\sum_{(i,j)\in Y^{(1)}}w\left(Y^{(2)}_{(i,j)}\right)
+\sum_{(i_n,j_n)\in Y^{(n)}}w\left(Y^{(3)}_{(i_2,j_2)}\right)w\left(Y^{(2)}_{(i_1,j_1)}\right)^*
+\sum_{(i,j)\in Y^{(1)}}w\left(Y^{(3)}_{(i,j)}\right)^*\\
-&\sum_{(i,j)\in Y^{(1)}}w\left(Y^{(2)}_{(i,j)}\right)w\left(Y^{(2)}_{(i,j)}\right)^*(1-r_{1,2})(1-r_{2,2})
-\sum_{(i,j)\in Y^{(2)}}w\left(Y^{(3)}_{(i,j)}\right)w\left(Y^{(3)}_{(i,j)}\right)^*(1-r_{1,3})(1-r_{2,3}).
\fe
Notice that this can be rewritten as 
\ie 
&\sum_{(i,j)\in Y^{(1)}}T\left(w\left(Y^{(2)}_{(i,j)}\right)\right)
+\sum_{(i,j)\in Y^{(2)}}T\left(w\left(Y^{(3)}_{(i,j)}\right)^*\right)
+\sum_{(i_n,j_n)\in Y^{(n)}}w\left(Y^{(3)}_{(i_2,j_2)}\right)w\left(Y^{(2)}_{(i_1,j_1)}\right)^*\\
-&\sum_{(i,j)\in Y^{(1)}}w\left(Y^{(2)}_{(i,j)}\right)^*r_{1,2}r_{2,2}
-\sum_{(i,j)\in Y^{(2)}}w\left(Y^{(3)}_{(i,j)}\right)^*r_{1,3}r_{2,3},
\fe
and $T\left(w\left(Y^{(2)}_{(i,j)}\right)\right)$, $T\left(w\left(Y^{(3)}_{(i,j)}\right)\right)$ have zero constant term in the limit $r_{1,2}r_{2,2}\to 1$, $r_{1,3}r_{2,3}\to 1$, as can be seen from \eqref{simple tangent character}. Thus, the constant term of $T^{\mathrm{vir}}_{s_1s_2\to 1}$ equals to the constant term of 
\ie 
\sum_{(i,j)\in Y^{(1)}}w\left(Y^{(2)}_{(i,j)}\right)^*_{r_{1,2}r_{2,2}\to 1}+\sum_{(i,j)\in Y^{(2)}}w\left(Y^{(3)}_{(i,j)}\right)^*_{r_{1,3}r_{2,3}\to 1}+\sum_{(i_n,j_n)\in Y^{(n)}}w\left(Y^{(3)}_{(i_2,j_2)}\right)w\left(Y^{(2)}_{(i_1,j_1)}\right)^*,
\fe
which is minus the combination of the number of non-empty Young diagrams $Y^{(2)}_{(i,j)}$, $Y^{(3)}_{(i,j)}$. In particular, it is negative if one of $Y^{(2)}_{(i,j)}$, $Y^{(3)}_{(i,j)}$ is non-empty. This concludes the proof of proposition.
\end{proof}
\noindent{\bf{The index of the open subset $\mathring{\cM}^{1,1,1}(M_1,M_2,M_3)$}}\\
\newline
We have seen in \eqref{projto3} that $\mathring{\cM}^{1,1,1}(M_1,M_2,M_3)$ is the vector bundle $\mathrm{Hom}(\cV_1,\cV_2)\otimes\mathrm{Hom}(\cV_2,\cV_3)\otimes\mathrm{Hom}(\cV_3,\cV_1)$, so its ring of functions is 
\ie 
\bigoplus_{k_2,k_3,k_1\ge 0}s_2^{k_2}s_3^{k_3}s_1^{k_1}S^{k_2}(\cV_2\otimes \cV_1^*)S^{k_3}(\cV_3\otimes \cV_2^*)S^{k_1}(\cV_1\otimes \cV_3^*)
\fe 
as a $\mathbf{T}_{\mathrm{edge}}$-equivariant sheaf on $\cM(M_1,1)\times \cM(M_2,1)\times \cM(M_3,1)$. We will use the following lemma. 
\begin{lemma}
\ie\label{decomposition of symmetric product}
S^{k_i}(\cV_{i}\otimes \cV_{i-1}^*)=\bigoplus_{|\underline{\lambda}|=k_i} S^{\underline{\lambda}}(\cV_i)\otimes S^{\underline{\lambda}}(\cV_{i-1}^*),
\fe 
where $\underline\lambda$ is a partition of $k_i$ (given by a Young diagram) and $S^{\underline{\lambda}}(\cV_i)$ is the irreducible representation of $\mathrm{GL}(\cV_i)$ defined $\mathrm{Hom}_{S_k}(R_{\underline\lambda},\cV_i^{\otimes k})$. Here $R_{\underline\lambda}$ is the irreducible representation of the permutation group $S_k$ with the Young diagram $\underline\lambda$. For example,$S^{(k_2)}(\cV_2)=S^{k_2}(\cV_2),\; S^{(1,1,\cdots, 1)}(\cV_2)=\wedge^{k_2}(\cV_2)$. 
\end{lemma}
\begin{proof}
Note that $S^k(\cV\otimes\cW)$ denotes the $S_k$ invariant of $(\cV\otimes\cW)^{\otimes k}$, which is a subspace of $(\cV\otimes\cW)^{\otimes k}$. Then consider
\ie
(\cV\otimes\cW)^{\otimes k}&=\cV^{\otimes k}\otimes\cW^{\otimes k}\\
&=\left(\bigoplus_{|\ld|=k}S^{\ld}(\cV)\otimes R_\ld\right)\otimes\left(\bigoplus_{|\ld'|=k}S^{\ld'}(\cW)\otimes R_{\ld'}\right),
\fe
where we used Schur-Weyl duality in the second line. $S^\ld(\bC^n)$ is a representation of $GL(n)$ labeled by a Young tableau $\ld$ and $R_\ld$ is a representation of $S_k$ labeled by $\ld$. We want the $S_k$-invariant part of $(\cV\otimes\cW)^{\otimes k}$ and $S_k$ to act only on $R_\ld$, $R_{\ld'}$. Using
\ie
\text{Hom}_{S_k}\left(\bC,R_\ld\otimes R_{\ld'}\right)=
\begin{cases} 0 &\text{if }R_\ld\neq (R_{\ld'})^*\\
\bC &\text{if }R_\ld=(R_{\ld'})^*
\end{cases},
\fe
we have
\ie
S^k(\cV\otimes\cW)=\bigoplus_{|\ld|=k}S^\ld(\cV)\otimes S^\ld(\cW).
\fe
\end{proof}
Using the decomposition \eqref{decomposition of symmetric product}, we can write the $\mathbf{T}_{\mathrm{edge}}$-equivariant K-theory index of $\mathring{\cM}^{1,1,1}(M_1,M_2,M_3)$ as
\ie 
&\chi(\mathring{\cM}^{1,1,1}(M_1,M_2,M_3),\mathcal O)=\sum_{\underline\lambda_1,\underline\lambda_2,\underline\lambda_3}s_1^{|\underline\lambda_1|}s_2^{|\underline\lambda_2|}s_3^{|\underline\lambda_3|}
\chi(\cM(M_1,1),S^{\underline\lambda_1}(\cV_1)\otimes S^{\underline{\lambda}_2}(\cV_1^*))\\
\times&~\chi(\cM(M_2,1),S^{\underline{\lambda}_2}(\cV_2)\otimes S^{\underline{\lambda}_3}(\cV^*_2))
~\chi(\cM(M_3,1),S^{\underline{\lambda}_3}(\cV_3)\otimes S^{\underline{\lambda}_1}(\cV^*_3)),
\fe 
where the sum is for all Young diagrams $\underline\lambda_i$, $i=1,2,3$. Note that
\ie \label{factorization of index}
\chi(\cM(M_i,1),S^{\underline{\lambda}}(\cV_i^*))=r_{1,i}^{-2M_i}r_{2,i}^{-2M_i}\chi(\cM(M_i,1),S^{\underline{\lambda}}(\cV_i))^*.
\fe 
Recall that the equivariant K-theory partition function of $\cM^{1,1,1}$ is defined as
\ie 
Z_{1,1,1}=\sum_{M_1,M_2,M_3}\chi(\cM^{1,1,1}(M_1,M_2,M_3),\mathcal O)x^{M_1}y^{M_2}z^{M_3}.
\fe 
Combining \eqref{factorization of index} with Proposition \ref{prop: the limit}, we then have
\begin{proposition}\label{prop: factorization}
In the limit $\underset{r_{1,1}r_{2,1}\to 1}{\lim}\:\underset{r_{1,2}r_{2,2}\to 1}{\lim}\:\underset{r_{1,3}r_{2,3}\to 1}{\lim}$, the equivariant K-theory partition function of $\cM^{1,1,1}$ factorizes as
\ie 
Z_{1,1,1}=&\sum_{\underline\lambda_i}s_1^{|\underline\lambda_1|}s_2^{|\underline\lambda_2|}s_3^{|\underline\lambda_3|}\mathcal Z^{\underline{\lambda_1},\underline{\lambda_2}}_1(r_{1,1},r_{2,1},x)_{r_{1,1}r_{2,1}\to 1}\mathcal Z^{\underline{\lambda_2},\underline{\lambda_3}}_1(r_{1,2},r_{2,2},y)_{r_{1,2}r_{2,2}\to 1}\\
\times&\mathcal Z^{\underline{\lambda_3},\underline{\lambda_1}}_1(r_{1,3},r_{2,3},z)_{r_{1,3}r_{2,3}\to 1}
\fe 
where $\mathcal Z^{\underline{\lambda}_i,\underline{\lambda}_{i+1}}_1(r_{1,i},r_{2,i},x_i)$ is the equivariant K-theory partition function for ADHM quiver:
\ie\label{halffactor}
\mathcal Z^{\underline{\lambda}_i,\underline{\lambda}_{i+1}}_1(r_{1,i},r_{2,i},x_i)=\sum _{M_i}\chi(\cM(M_i,1),S^{\underline{\lambda}_i}(\cV_i)\otimes S^{\underline{\lambda}_{i+1}}(\cV^*_i))x_i^{M_i},
\fe 
where $x_1=x$, $x_2=y$, $x_3=z$.
\end{proposition}

\begin{remark}
Consider the ADHM quiver with the framing rank one and the gauge rank $M_1$, and the additional framing nodes with the rank $M_2$, $M_3$, then it is easy to see that the corresponding quiver variety is the universal bundle $\cV$ in the ADHM quiver without extra framing. Note that there is an edge torus action on the quiver variety and the equivariant K-theory index for this variety is 
\ie 
\chi(\cM(M_1,1),S^{\underline{\lambda}_1}(\cV_1)\otimes S^{\underline{\lambda}_2}(\cV^*_1))s_1^{|\ld_1|}
\fe
This is exactly the coefficient in the summation \eqref{halffactor}.
\end{remark}
\subsection{Connection to the 5d Chern-Simons algebra of operators}\label{sec:connectionto5dcs}

Notice that in the $s_i\to 0, r_{1_i},r_{2,i}\to 1$ limit, the equivariant K-theory partition function of $\cM^{1,1,1}$ equals 
\ie 
\sum_{M_1,M_2,M_3}\chi(\cM(M_1,1),\mathcal O)\chi(\cM(M_2,1),\mathcal O)\chi(\cM(M_3,1),\mathcal O)x^{M_1}y^{M_2}z^{M_3},
\fe 
i.e. it is a product of equivariant K-theory partition functions of ADHM quivers:
\ie 
\lim_{s_i\to 0}Z_{1,1,1}=Z_{\mathrm{ADHM}}(r_{1,1},r_{2,1},x)_{r_{1,1}r_{2,1}\to 1}Z_{\mathrm{ADHM}}(r_{1,2},r_{2,2},y)_{r_{1,2}r_{2,2}\to 1}Z_{\mathrm{ADHM}}(r_{1,3},r_{2,3},z)_{r_{1,3}r_{2,3}\to 1}.
\fe 
Since the equivariant K-theory of ADHM quiver varieties admits a structure of Verma modules of the affine quantum group $U_q(\hat{\mathfrak{gl}}_1)$, we see that the $s_i=0$ part of $Z_{1,1,1}$ is the character of a Verma module for $U_q(\hat{\mathfrak{gl}}_1)\otimes U_q(\hat{\mathfrak{gl}}_1)\otimes U_q(\hat{\mathfrak{gl}}_1)$. After passing from K-theory to cohomology, part of $Z_{1,1,1}$ is a character of a Verma module for the affine Yangian of $\mathfrak{gl}_1\oplus \mathfrak{gl}_1\oplus \mathfrak{gl}_1$, which is exactly the algebra of gauge-invariant local observables of the 5d holomorphic Chern-Simons theory with gauge group $U(1)\times U(1)\times U(1)$. 
\begin{conjecture}
There is an action of $Y(\hat{\mathfrak{gl}}_1)\otimes Y(\hat{\mathfrak{gl}}_1)\otimes Y(\hat{\mathfrak{gl}}_1)$ on
\ie 
\bigoplus_{M_1,M_2,M_3}H^*_{\mathbf{T}_{\mathrm{edge}}}(\cM^{1,1,1}(M_1,M_2,M_3)).
\fe
\end{conjecture}
\noindent It will be nice to find the map between the fugacities of $\mathbf{T}_{\mathrm{edge}}$ and $\epsilon_1,\epsilon_2,\epsilon_3$ in $Y(\hat{\mathfrak{gl}}_1)$.

\section{Gluing two partition functions}\label{sec:5}
In this section, we discuss the glueing of the building blocks $\cM^{1,1}(M_1,M_2)$\footnote{We thank Michele Del Zotto for suggesting to consider gluing.}. Consider $n+1$ copies of ADHM quivers, with framing rank $1$ and gauge ranks $M[0],\cdots,M[n]$, then we connect these ADHM quivers by joining $M[i]$ with $M[i+1]$ by arrow $S_i$. When $n=1$ this is the quiver configuration for $\cM^{1,1}(M[0],M[1])$, and for $n>1$ this new quiver can be thought of as the gluing of $n$ copies of the above building blocks by identifying the gauge nodes. We call the moduli space of stable representation of the the above quiver $\cM(M[0],\cdots,M[n])$.

For each gauge node $i$, there is a torus $\mathbf{T}'_i=\mathbb C^{\times 2}$ acting on the quiver data as
\ie
(X_i,Y_i,I_i,J_i)\mapsto (a_iX_i,b_iY_i,I_i,a_ib_iJ_i),\quad (a_i,b_i)\in \mathbf{T}'_i.
\fe
There is another set of tori $\tilde{\mathbf{T}}_i=\bC^{\times}$ acting on the quiver data by scaling the $S_i$ arrow:
\ie
S_i\mapsto c_iS_i,\quad c_i\in \tilde{\mathbf{T}}_i.
\fe
It is obvious that the actions of $\mathbf{T}'_i$ and $\mathbf{S}_i$ commute with the gauge group, and thus the torus $\mathbf{T}_{\mathrm{edge}}=\prod_{i=0}^n\mathbf{T}'_i\times \prod_{i=0}^{n-1}\tilde{\mathbf{T}}_i$ acts on the moduli space $\cM(M[0],\cdots,M[n])$. Moreover, it is easy to see that $\cM(M[0],\cdots,M[n])^{\mathbf{T}_{\mathrm{edge}}}$ is a set of disjoint union of reduced points. Therefore it makes sense to discuss the $\mathbf{T}_{\mathrm{edge}}$-equivariant K-theory class of $\chi(\cM(M[0],\cdots,M[n]),\cO)$.

Let us denote the equivariant parameters of $\mathbf{T}'_i$ by $r_{1,i},r_{2,i}$, and the equivariant parameters of $\tilde{\mathbf{T}}_i$ by $s_{i}$. The following is analogous to Proposition \ref{prop: the limit}.
\begin{proposition}
In the limit $\underset{r_{1,0}r_{2,0}\to 1}{\lim}\:\underset{r_{1,1}r_{2,1}\to 1}{\lim}\cdots\underset{r_{1,n}r_{2,n}\to 1}{\lim}$, we have 
\ie 
\chi(\cM(M[0],\cdots,M[n]),\cO)_{r_{1,i}r_{2,i}\to 1}=\chi(\mathring{\cM}(M[0],\cdots,M[n]),\mathcal O)_{r_{1,i}r_{2,i}\to 1},
\fe
where $\mathring{\cM}(M[0],\cdots,M[n])$ is the open subset of $\cM(M[0],\cdots,M[n])$ such that the condition
$$\forall i,\;\bC\langle X_i,Y_i\rangle \mathrm{Im}(I_i)=\bC^{M[i]},$$is satisfied.
\end{proposition}

The proof is omitted since it is similar to that of Proposition \ref{prop: the limit}. Note that the the moduli space $\mathring{\cM}(M[0],\cdots,M[n])$ is a vector bundle over $\prod_{i=0}^n \cM(M[i],1)$ with the fiber
\ie 
\prod_{i=0}^{n-1}\mathrm{Hom}(\cV_i,\cV_{i+1}),
\fe 
where $\cV_i$ is the tautological bundle on $\cM^1(M[i])$. Therefore, we can easily write down the $\mathbf{T}_{\mathrm{edge}}$-equivariant K-theory partition function
\ie 
Z^{(n)}=\sum_{M[0],\cdots,M[n]}\chi(\mathring{\cM}(M[0],\cdots,M[n]),\cO)x_0^{M[0]}\cdots x_n^{M[n]}.
\fe
Using Lemma \ref{decomposition of symmetric product}, we can write fiber bundle $\prod_{i=0}^{n-1}\mathrm{Hom}(\cV_i,\cV_{i+1})$ as direct sum of symmetric powers:
\ie 
\bigoplus_{\underline{\lambda}_0,\cdots ,\underline{\lambda}_{n-1}} S^{\underline{\lambda}_0}(\cV_0)\otimes S^{\underline{\lambda}_0}(\cV_1^*)\otimes S^{\underline{\lambda}_1}(\cV_1)\otimes S^{\underline{\lambda}_1}(\cV_2^*)\otimes \cdots\otimes S^{\underline{\lambda}_{n-1}}(\cV_{n-1})\otimes S^{\underline{\lambda}_{n-1}}(\cV_{n}^*),
\fe 
where the summation is over Young diagrams $\underline{\lambda}_0,\cdots ,\underline{\lambda}_{n-1}$. Thus the partition function $Z$ decomposes as
\ie 
Z^{(n)}=\sum_{\underline{\lambda}_0,\cdots ,\underline{\lambda}_{n-1}}& s_0^{|\underline{\lambda}_0|}\cdots s_n^{|\underline{\lambda}_n|}\cZ_1^{\underline{\lambda}_0,\emptyset}(r_{1,0},r_{2,0},x_0)\cZ_1^{\underline{\lambda}_1,\underline{\lambda}_0}(r_{1,1},r_{2,1},x_1)\cdots\\
&\cdots\cZ_1^{\underline{\lambda}_{n-1},\underline{\lambda}_{n-2}}(r_{1,n-1},r_{2,n-1},x_{n-1})\cZ_1^{\emptyset,\underline{\lambda}_{n-1}}(r_{1,n},r_{2,n},x_n),
\fe 
where $\mathcal Z^{\underline{\lambda},\underline{\mu}}_1(r_{1},r_{2},x)$ is the equivariant K-theory partition function for the ADHM quiver:
\ie\label{halffactor}
\mathcal Z^{\underline{\lambda},\underline{\mu}}_1(r_{1},r_{2},x)=\sum _{M}\chi(\cM(M,1),S^{\underline{\lambda}}(\cV)\otimes S^{\underline{\mu}}(\cV^*))x^{M}.
\fe 
Using the above factorization, the partition function $Z(n)$ can be built up from the basic blocks $Z^{(1)}$:
\ie 
Z^{(n)}=\sum_{\underline{\lambda}_0,\cdots ,\underline{\lambda}_{n-1}} s_0^{|\underline{\lambda}_0|}\cdots s_n^{|\underline{\lambda}_n|}&Z_{\underline{\lambda}_0}^{(1)}(r_{1,0},r_{2,0},r_{1,1},r_{2,1},x_0,x_1)Z_{\underline{\lambda}_{n-1}}^{(1)}(r_{1,n-1},r_{2,n-1},r_{1,n},r_{2,n},x_{n-1},x_n)\\
&\times F_{\underline{\lambda}_0,\underline{\lambda}_1}(r_{1,1},r_{2,1},x_1) \cdots F_{\underline{\lambda}_{n-2},\underline{\lambda}_{n-1}}(r_{1,n-1},r_{2,n-1},x_{n-1}).
\fe 
Here $Z_{\underline{\lambda}_i}^{(1)}(r_{1,i},r_{2,i},r_{1,i+1},r_{2,i+1},x_i,x_{i+1})$ is the coefficient of $s_i^{|\underline{\lambda}_i|}$ in the expansion
$$Z^{(1)}=\sum_{\underline{\lambda}_i} s_i^{|\underline{\lambda}_i|}Z_{\underline{\lambda}_i}^{(1)}(r_{1,i},r_{2,i},r_{1,i+1},r_{2,i+1},x_i,x_{i+1}).$$The gluing factor $F_{\underline{\lambda}_{i-1},\underline{\lambda}_{i}}(r_{1,i},r_{2,i},x_i)$ is
\ie 
F_{\underline{\lambda}_{i-1},\underline{\lambda}_{i}}(r_{1,i},r_{2,i},x_i)=\frac{\cZ^{\underline{\lambda}_{i},\underline{\lambda}_{i-1}}_1(r_{1,i},r_{2,i},x_i)}{\cZ^{\underline{\lambda}_{i},\emptyset}_1(r_{1,i},r_{2,i},x_i)\cZ^{\emptyset,\underline{\lambda}_{i-1}}_1(r_{1,i},r_{2,i},x_i)}
\fe 

\section*{Acknowledgments}
We are grateful to Michele Del Zotto for his collaboration at the early stage of this project and for valuable comments and suggestions to improve the draft. We thank Christopher Beem, Kevin Costello, Max H\"{u}bner, Sakura Schafer-Nameki for useful discussion. A part of this work was done when we were participating in the Aspen Winter Workshop. The authors thank the organizers of the workshop and the Aspen Center for their warm hospitality. Research of JO was supported by ERC grants 864828 and 682608. Research at the Perimeter Institute is supported by the Government of Canada through Industry Canada and the Province of Ontario through the Ministry of Economic Development $\&$ Innovation.
\providecommand{\href}[2]{#2}\begingroup\raggedright
    
    \endgroup

\end{document}